\documentclass{article}

\usepackage{adjustbox}
\usepackage[ruled,vlined,noend]{algorithm2e}
\usepackage{amsfonts,amsmath,amssymb,amsthm}
\usepackage{authblk}
\usepackage{bbm}
\usepackage{caption}
\usepackage{float}
\usepackage[edges]{forest}
\usepackage{geometry}
\usepackage[colorlinks = true,
            linkcolor = blue,
            urlcolor = blue,
            citecolor = teal,
            anchorcolor = blue, 
            breaklinks = true]{hyperref}
\usepackage{listings}
\usepackage{mathtools}
\mathtoolsset{showonlyrefs=false}
\usepackage{natbib}
\usepackage{pgfplots}
\usepackage{physics}
\usepackage{standalone}
\usepackage{subcaption}
\usepackage{tikz}
\usepackage{xcolor}

\usepgfplotslibrary{groupplots}
\pgfplotsset{compat=1.14}

\newcommand*{\bbE}{\mathbb{E}}  %
\newcommand*{\bbM}{\mathbb{M}}  %
\newcommand*{\bbP}{\mathbb{P}}  %
\newcommand*{\bbQ}{\mathbb{Q}}  %
\newcommand*{\bbR}{\mathbb{R}}  %

\newcommand*{\Chi}{\mathcal{X}}  %
\newcommand*{\calY}{\mathcal{Y}}  %

\newtheorem{corollary}{Corollary}[section]
\newtheorem{definition}{Definition}[section]

\newtheorem{proposition}{Proposition}[section]
\newtheorem{remark}{Remark}[section]

\SetKwFor{ParFor}{for}{in parallel do}{end}
\SetKwBlock{Locals}{Local Variables:}{end}
\SetKwProg{Fn}{Function}{}{}
\DontPrintSemicolon

\DeclarePairedDelimiter\floor{\lfloor}{\rfloor}
\DeclarePairedDelimiter\ceil{\lceil}{\rceil}

\newcommand{\ft}[3]{#1_{#2:#3}}  %

\newcommand{\MC}{Monte Carlo}

\newcommand{\bigO}{\mathcal{O}}

\usetikzlibrary{
    arrows.meta,
    backgrounds,
    calc,
    fit,
    positioning}

\tikzstyle{state}=[circle,
                   thick,
                   draw=blue!80,
                   fill=blue!20]

\tikzstyle{operation}=[rectangle,
                       thick,
                       draw=gray!80,
                       fill=gray!20]

\tikzstyle{background}=[rectangle,
                        fill=gray!0,
                        inner sep=0.2cm,
                        rounded corners=5mm]
                        
\definecolor{color0}{HTML}{332288}
\definecolor{color1}{HTML}{88CCEE}
\definecolor{color2}{HTML}{44AA99}
\definecolor{color3}{HTML}{117733}
\definecolor{color4}{HTML}{999933}
\definecolor{color5}{HTML}{DDCC77}

\title{De-Sequentialized Monte Carlo: a parallel-in-time particle smoother}
\author[1]{Adrien Corenflos}
\author[2]{Nicolas Chopin}
\author[1]{Simo S\"{a}rkk\"{a}}
\affil[1]{Department of Electrical Engineering and Automation, Aalto University.}
\affil[2]{ENSAE, Institut Polytechnique de Paris.}

\date{\today}

\begin{document}
\maketitle

\begin{abstract}
    Particle smoothers are SMC (Sequential Monte Carlo) algorithms designed to
    approximate the joint distribution of the states given observations from a
    state-space model. 
    We propose dSMC (de-Sequentialized Monte Carlo), a new particle smoother that is
    able to process $T$ observations in $\bigO(\log T)$ time on parallel
    architecture. This compares favourably with standard particle smoothers,
    the complexity of which is linear in $T$. We derive $\mathcal{L}_p$ convergence
    results for dSMC, with an explicit upper bound, polynomial in $T$. 
    We then discuss how to reduce the variance of the smoothing estimates computed by dSMC by (i) 
    designing good proposal distributions for sampling the particles at the initialization of the algorithm, as well as by (ii) using 
    lazy resampling to increase the number of particles used in dSMC. 
    Finally, we design a particle Gibbs sampler based on dSMC, which is able to perform parameter inference in a state-space
    model at a $\bigO(\log(T))$ cost on parallel hardware.
\end{abstract}

{\bf Keywords:} Sequential Monte Carlo; Parallel methods; Particle filtering; Particle Gibbs; Particle smoothing

\section{Introduction}
\label{sec:intro}

State space models (SSM), or hidden Markov models, are a class of statistical models that comprise unobserved (latent) Markovian states $X_t \in \Chi{}$ for $t \in \{0, 1, \ldots, T\}$, and conditionally independent observations $Y_t \in \calY{}$ \citep[see, e.g., ][]{Jazwinski:1970,Cappe+Moulines:2005,Sarkka2013Book,Chopin2020Book}. The models can be written in the form
\begin{align}
  X_t \mid x_{t-1} &\sim P_t(\dd{x_t} \mid x_{t-1}), \\
  Y_t \mid x_{t} &\sim P_t(\dd{y_t} \mid x_{t}),
\label{eq:ssm1}
\end{align}
with $X_0 \sim \bbP_0(\dd{x_0})$, where $P_t(\dd{x_t} \mid x_{t-1})$ is the transition kernel of the Markov sequence $X_t$ modeling the dynamics of the system, $P_t(\dd{y_t} \mid x_{t})$ is the conditional distribution of measurements $Y_t$, and $\bbP_0(\dd{x_0})$ is the prior distribution of the initial state $X_0$. For simplicity, we assume that there exist $h_t$ such that $P_t(\dd{y_t} \mid x_{t}) = h_t(y_t \mid x_t) \dd{y_t}$, where $\dd{y_t}$ refers to a dominating measure over $\calY{}$, for example, the Lebesgue measure if $\calY{}=\bbR^{d_y}$. 

In this paper, we consider the state-estimation problem which refers to the problem of inferring the states $X_t$ from the measurements $Y_t$. In particular, we concentrate on the smoothing problem \citep{Wiener:1949,Sarkka2013Book,Chopin2020Book}, where the aim is to infer the distribution of the whole trajectory of states $\ft{X}{0}{T}$ given the whole trajectory of measurements $\ft{Y}{0}{T}$. A typical application of state-space estimation consists in target tracking, where the state $X_t$ models the position (and possibly other physical quantities such as the speed) of a moving target, and the observation $Y_t$ corresponds to some noisy and partial or indirect measurement of $X_t$ \citep{Jazwinski:1970,Bar-Shalom+Li+Kirubarajan:2001}. Additionally, we consider the parameter-estimation problem of inferring the unknown parameters appearing in the model. In addition to target tracking, SSMs state and/or parameter estimation problems also arise in various other applications such as in biomedicine, epidemiology, finance, audio signal analysis, and imaging~\citep[for a review, see, e.g.][]{Sarkka2013Book,Chopin2020Book}.

In the signal processing setting \citep[e.g.][]{Jazwinski:1970,Sarkka2013Book} the solutions to the smoothing problem are focused on computing the marginal conditional distributions of $X_t$ for $t = 0,\ldots,T$ given all the measurements $\ft{Y}{0}{T}$. However, in the context of Monte Carlo methods \citep{Chopin2020Book} -- which we also concentrate on here -- it is more natural to directly consider the joint distribution of all the states and measurements which can be written as
\begin{align}\label{eq:ssm}
    \bbP(\dd{\ft{x}{0}{T}}, \dd{\ft{y}{0}{T}})
        \coloneqq \bbP_0(\dd{x_0}) \left\{\prod_{t=0}^T h_t(y_t \mid x_t) \dd{y_t}\right\} \left\{\prod_{t=1}^T P_t(\dd{x_t} \mid x_{t-1}) \right\}.
\end{align}
In this notation, smoothing consists in representing the posterior distribution of the states conditionally on the observations $\bbQ_T(\dd{\ft{x}{0}{T}}) \coloneqq \bbP(\dd{\ft{x}{0}{T}} \mid \ft{y}{0}{T})$ and, in particular, being able to approximate expectations such as $\bbQ_T(\varphi) \coloneqq \bbE_{\bbQ_T}\left[\varphi(\ft{X}{0}{T})\right]$ for some function $\varphi$ of interest. When $\bbP_0(\dd{x_0}) = \bbP_0(\dd{x_0} \mid \theta)$, $P_t(\dd{x_t} \mid x_{t-1}) = P_t(\dd{x_t} \mid x_{t-1}, \theta)$, and $h_t(y_t \mid x_{t}) = h_t(y_t \mid x_{t}, \theta)$ depend on a parameter $\theta$, parameter estimation consists in computing estimates $\theta$ either as point estimates or in form of a posterior distribution of the parameter. Formally, if $\theta$ is given a prior distribution $p(\dd{\theta})$, one can represent its posterior distribution as
\begin{align}
  p(\dd{\theta} \mid \ft{y}{0}{T}) \propto p(\dd{\theta}) p(\ft{y}{0}{T} \mid \theta),
\end{align}
where
\begin{align}
    p(\ft{y}{0}{T} \mid \theta) = \int_{\Chi^T} \bbP_0(\dd{x_0} \mid \theta) \left\{\prod_{t=0}^T h_t(y_t \mid x_t, \theta)\right\} \left\{\prod_{t=1}^T P_t(\dd{x_t} \mid x_{t-1}, \theta) \right\}.
\end{align}

Except in the case of finite-state SSMs \citep[e.g.][]{rabiner1989tutorial}, linear Gaussian SSMs (LGSSMs) \citep{kalman1960KF,Rauch1965RTS}, and certain other special cases, neither the smoothing nor the parameter estimation problems admit a closed-form solution, and we need to resort to approximations. A successful class of such approximations comprise Gaussian approximation based filtering and smoothing approximations such as extended \citep{Jazwinski:1970}, unscented \citep{Julier+Uhlmann+Durrant-Whyte:2000}, and cubature Kalman filters \citep{Ito+Xiong:2000,Arasaratnam+Haykin:2009}, as well as their corresponding smoothers \citep[for a review, see, e.g.,][]{Sarkka2013Book}. Another class of methods is sequential Monte Carlo (SMC) algorithms \citep[see, e.g,][]{Gordon+Salmon+Smith:1993,Doucet+Godsill+Andrieu:2000, Chopin2020Book} such as particle filters and smoothers which are based on Monte Carlo sampling from the filtering and smoothing distributions. These algorithms can, more generally, also sample from the full distribution of models given as a product of Markov kernels and potentials $h_t$
\begin{align}
\label{eq:fk}
    \bbQ_T(\dd{\ft{x}{0}{T}})
        \propto \bbP_0(\dd{x_0}) \left\{\prod_{t=0}^T h_t(x_t)\right\} \left\{\prod_{t=1}^T P_t(\dd{x_t} \mid x_{t-1}) \right\}.
\end{align}
which recover case of \eqref{eq:ssm} by setting $h_t(x_t) = h_t(y_t \mid x_t)$ in a slight abuse of notation. 

The aforementioned finite-state methods, Gaussian approximations, and SMC methods are based on sequential forward and backward recursions which allow for computationally efficient algorithms which scale linearly in the number of time steps $\bigO(T)$. Although this computational complexity is (in a sense) optimal in classical single-core computers, it is not optimal in multi-core parallel computers which are capable of sub-linear time-complexity in terms of span-complexity \citep{Cormen:2009}. Here span-complexity refers to the actual wall-clock time taken by a method when run on a parallel computer which can be less than $\bigO(T)$ even when the size of data is $T$. The sequential approximations for filtering and smoothing, in their standard formulation, have a linear time complexity in $T$ even when run on a parallel computer, which is due to the inherent sequential nature of the computations.

However, it was recently shown in \citet{Sarkka2021ParallelKF} that Bayesian filtering and smoothing recursions (including, e.g., the Kalman filter and smoother) can be reformulated in terms of associative operators that can be time-parallelized to $\bigO(\log T)$ span-complexity by using a parallel scan algorithm. In \citet{Hassan:2021}, similar methods were developed for finite-state models, and \citet{Yaghoobi2021ParallelEKF} developed Gaussian approximation based parallel methods for non-linear SSMs. These methods reduce the computational cost from linear to logarithmic in the number of observations on highly parallel hardware such as graphics processing units (GPUs). Unfortunately, the general formulation of \citet{Sarkka2021ParallelKF} is not directly applicable to SMC-based particle filters and smoothers, as propagating the associative operator appearing in \citet{Sarkka2021ParallelKF} is exactly what SMC offers to do in the first place. The aim of this article is to fix this shortcoming by proposing a parallel-in-time formulation of SMC, the de-Sequentialized Monte Carlo (dSMC) method, that can be used -- either as a standalone method, or in combination with Gaussian approximations -- in order to perform Monte Carlo inference in general SSMs. However, instead of using an associative operator formulation as in \citet{Sarkka2021ParallelKF}, the method uses parallel merging of blocks in a tree structure.

\subsection{Related work}
    \label{subsec:related-work}
    Temporal parallelization of general Bayesian filters and smoothers have recently been discussed in \citet{Sarkka2021ParallelKF}, \citet{Hassan:2021}, and \citet{Yaghoobi2021ParallelEKF}, but only in the contexts of Gaussian approximations and finite-state models. Parallelization methods for Kalman type of (ensemble) filters via parallel matrix computations over the state dimension are presented in \citet{Lyster:1997} and \citet{Evensen:2003}. In the context of SMC methods, parallelization over particles has been considered in \citet{Lee:2010,Rosen:2013,Murray2016ParallelResampling}, however, these methods do not address the time dimension and their computational complexity is still linear in $T$ on parallel hardware. In the context of variational inference~\citep[see, e.g.][]{blei2017variational}, it was also noted in \citet{aitchison2019tensor} that operations akin to sequential importance sampling could be easily written as chaining matrix multiplications, allowing to parallelize these on a GPU, both in the time and particle dimensions. The work of~\citet{SinghBlockedGibbs2017} considers blocking strategies for particle Gibbs algorithm, using the Markov property to allow the treatment of non-contiguous time blocks in parallel. Their method however works better for larger blocks, with significant overlap, thereby reducing its parallelization properties, they also do not consider parallelization of particle smoothing. Orthogonally to these direction, coupled smoothing methods, introduced in \citet{jacob2019smoothing} and further developed in \citet{middleton2019unbiased,lee2020coupled}, allow to run compute unbiased estimates of particle smoothers. This allows to parallelize calculation of smoothing expectations by aggregating many unbiased smoothers together.

Closest to our work is \citet{Lindsten2017Divideandconquer} which considers the case of already formed graphical models. In fact, once the tree structure of dSMC is built, our algorithm operates similarly to the divide and conquer SMC algorithm of \cite{Lindsten2017Divideandconquer}, which propagates and merges particle samples from children nodes to a parent node. In their article, \citet{Lindsten2017Divideandconquer} show the consistency of their algorithm in terms of convergence in probability. This was further improved by \citet{kuntz2021divide} who derived additional theoretical properties of estimates computed from divide-and-conquer SMC. These results can be applied to
dSMC as well. However, our method differs from both these articles in several ways. First, \citet{Lindsten2017Divideandconquer} do not consider \emph{modifying} the structure of a pre-existing graphical model to be able to parallelize it. Second, the bounds for $\mathcal{L}_p$ errors we derive in this article depend explicitly (and polynomially) in $T$. These results are specific to dSMC as a parallel algorithm. Third, we derive a parallel-in-time particle Gibbs algorithm for dSMC which can be more generally applied to \citet{Lindsten2017Divideandconquer}. Lastly, we introduce parallel-in-time initialization of the algorithm and lazy resamplings as a way to speed up the algorithm and allow for better scalability in the number of particles used. 

Finally, we note that \citet{ding2018treebased} introduced a smoothing algorithm leveraging the same binary tree. However, their method differs from ours in several aspects. (i) The main goal of \citet{ding2018treebased} is to reduce the variance of smoothing algorithms by computing adapted target distributions at each node of the tree. (ii) As a consequence, they do not directly address parallelization in time (our main motivation), and, in fact, do not allow for it as their algorithm requires to run a particle filter and a particle smoother a priori. They also (iii) do not discuss approximated LGSSM PIT initialization, lazy schemes, or particle Gibbs extensions.

\subsection{Contributions}
    \label{subsec:contribution}
    In Section~\ref{sec:method}, we introduce a formal divide-and-conquer formulation of the smoothing distribution for a class of Feynman-Kac models, which is then used to define dSMC. We then proceed to study the properties of dSMC, in particular, we derive $\mathcal{L}_p$ error bounds that only scale polynomially in $T$ for balanced tree representations of the smoothing distribution. Section~\ref{sec:pgibbs} is concerned with introducing the conditional formulation of dSMC. This is then used to define a PIT particle Gibbs algorithm. In Section~\ref{sec:proposals}, we discuss how to construct adapted proposals without breaking the logarithmic scaling in $T$, and then show how parallel resampling methods can be used to lazily increase the number of particles used in dSMC. Finally, in Section~\ref{sec:experiments} we experimentally demonstrate the statistical and computational properties of our method on a suite of examples. The article concludes with a discussion of the limitations and possible improvements of the de-Sequentialized \MC{} method.

\section{De-Sequentialized Monte Carlo}
\label{sec:method}
We first introduce the core components required for building a parallel-in-time (PIT) particle smoother algorithm, that we coin de-Sequentialized Monte Carlo (dSMC). Our method relies on a divide-and-conquer approach, where we recursively stitch together partial smoothing distributions $\ft{\bbQ}{a}{b}(\dd{\ft{x}{a}{b}})$ in order to form the final estimate. In order to do this, we first present the tree structure associated with smoothing in state-space models, then we discuss how importance sampling-resampling can be leveraged to create joint samples from marginal ones. Finally, we describe the resulting algorithm and derive convergence bounds for it. For the sake of generality, we will consider the potential formulation $h_t(x_t)$ in \eqref{eq:fk}, which possibly depends on $y_t$, but, by a slight abuse of language, we will still refer to $\bbQ_T$ as the smoothing distribution.

\subsection{Tree structure}\label{subsec:tree-structure}
    The recursive expressions for the smoothing distribution
\begin{align}
    \bbQ_{T}(\dd{\ft{x}{0}{T}}) = \frac{1}{L_T} \left[\bbP_0(\dd{x_0}) \prod_{t=1}^T P_t(\dd{x_t} \mid x_{t-1})\right] \prod_{t=0}^T h_t(x_t), 
\end{align}
where $L_T $ is a normalizing constant, 
are given by the forward Feynman-Kac recursion~\citep[see, e.g.][]{DelMoral2004Book, Chopin2020Book}
\[\bbQ_{t+1}(\dd{\ft{x}{0}{t+1}})
  \propto
  \bbQ_{t}(\dd{\ft{x}{0}{t}})
  h_{t+1}(x_{t+1}) P_t(\dd{x_{t+1}} \mid x_t),
\]  
or the backward one
\[\bbQ_{T}(\dd{\ft{x}{t}{T}}) \propto h_{t+1}(x_{t+1}) p_t(x_{t+1} \mid x_t) \dd{x_t} \bbQ_{T}(\dd{\ft{x}{t+1}{T}}),\]
when $P_t(\dd{x_{t+1}} \mid x_t)$ admits a density $p_t(x_{t+1} \mid x_t)$ with
respect to a fixed ($x_t$-independent) measure $\dd{x}_{t+1}$. 
Leveraging these recursions respectively corresponds to particle filtering
and particle smoothing algorithms, and results in algorithms for sampling from
$\bbQ_{0:T}$ that scale computationally in $\bigO(T)$.

In this section we instead propose a divide-and-conquer recursive construction
of the smoothing density $\bbQ_T$. In order to do so, we introduce the concept
of partial smoothing distributions. 
\begin{definition}\label{def:twisted-model}
    Let $(\nu_c(\dd{x_c}))_{c=0}^T$ be a collection of probability measures,
    such that for all $c> 0$ and all $x_{c-1} \in \Chi{}$, $P_c(\dd{x_c}| x_{c-1})$ 
    is absolutely continuous with respect to $\nu_c(\dd{x_c})$. Then
    for any $0 \leq a \leq b \leq T$, we can define
    \begin{align}
        \ft{\bbQ^{\nu}}{a}{b}(\dd{\ft{x}{a}{b}}) \coloneqq
        \frac{1}{\ft{L^{\nu}}{a}{b}} \left[\nu_{a}(\dd{x_a}) \prod_{t=a+1}^b
        P_t(\dd{x_t} \mid x_{t-1}) \right] \prod_{t=a+1}^b h_t(x_t),
    \end{align}
    where $\ft{L^{\nu}}{a}{b}$ is a normalizing constant (assumed to be
    positive), and by convention the product over an empty set is $1$, so that,
    for any $a$, $\ft{\bbQ^{\nu}}{a}{a}(\dd{x_a}) = \nu_a(\dd{x}_a)$ and
    $L^{\nu}_{a:a} = 1$.
\end{definition}
Provided that $\nu_0$ defines the filtering posterior of $x_0$, we can then
recover the original $\ft{\bbQ}{0}{T}$ from  $\ft{\bbQ^{\nu}}{0}{T}$.
\begin{proposition}\label{prop:reweighting} For any family $\ft{\nu}{0}{T}$
    given by Definition \ref{def:twisted-model}, and such that 
    $\nu_0(\dd{x_0}) \propto h_0(x_0) \bbP_0(\dd{x_0})$, we have
    \begin{align}\label{eq:reweighting} \ft{\bbQ}{0}{T}(\dd{\ft{x}{0}{T}}) 
    = \ft{\bbQ^{\nu}}{0}{T}(\dd{\ft{x}{0}{T}}).
    \end{align}
\end{proposition}

The partial smoothing distributions $\left(\ft{\bbQ^{\nu}}{a}{b}\right)_{0 \leq a < b, \leq T}$ can then be stitched together, forming a recursive structure for the smoothing operation. 
\begin{proposition}\label{prop:stitching}
     For any $0 \leq a < c < b \leq T$, we have
     \begin{align}\label{eq:stitching}
        \ft{\bbQ^{\nu}}{a}{b}(\dd{\ft{x}{a}{b}}) 
            &=
            \frac{\ft{L^{\nu}}{a}{c-1} \ft{L^{\nu}}{c}{b} }
                 {\ft{L^{\nu}}{a}{b}}
            \omega^{\nu}_c(x_{c-1}, x_c) 
            \ft{\bbQ^{\nu}}{a}{c-1}(\dd{\ft{x}{a}{c-1}}) 
            \ft{\bbQ^{\nu}}{c}{b}(\dd{\ft{x}{c}{b}}),
     \end{align}
     where $\omega^{\nu}_c$ is defined as the following Radon–Nikodym derivative:
     \begin{align}
         \omega^{\nu}_c(x_{c-1}, x_c) \coloneqq \frac{P_c(\dd{x_c} \mid x_{c-1}) h_c(x_c)}{\nu_c(\dd{x_c})}. \label{eq:weights-formula}
     \end{align}
\end{proposition}
\begin{proof}
    For all $0\leq a < c < b \leq T$, we have:
    \begin{align}
        &\quad \frac{\ft{L^{\nu}}{a}{c-1} \ft{L^{\nu}}{c}{b}}{\ft{L^{\nu}}{a}{b}}
         \omega^{\nu}_c(x_{c-1}, x_c) 
         \ft{\bbQ^{\nu}}{a}{c-1}(\dd{\ft{x}{a}{c-1}}) 
         \ft{\bbQ^{\nu}}{c}{b}(\dd{\ft{x}{c}{b}})\\
        &= \frac{1}{\ft{L^{\nu}}{a}{b}} \frac{P_c(\dd{x_c}\mid x_{c-1}) h_c(x_c)}{\nu_c(\dd{x_c})} \\
        &\times \left[\nu_{a}(\dd{x_a}) \prod_{t=a+1}^b P_t(\dd{x_t} \mid x_{t-1}) \right] \prod_{t=a+1}^{c-1} h_t(x_t) \\
        &\times \left[\nu_{c}(\dd{x_c}) \prod_{t=c+1}^{b} P_t(\dd{x_t} \mid x_{t-1}) \right] \prod_{t=c+1}^b h_t(x_t)\\
        &=\frac{1}{\ft{L^{\nu}}{a}{b}} \left[\nu_{a}(\dd{x_a}) \prod_{t=a+1}^b P_t(\dd{x_t} \mid x_{t-1}) \right] \prod_{t=a+1}^{b} h_t(x_t)\\
        &=\ft{\bbQ^{\nu}}{a}{b}(\dd{\ft{x}{a}{b}}).
    \end{align}
\end{proof}
The recursive property exhibited by Proposition~\ref{prop:stitching} allows us to construct an arbitrary tree structure on the smoothing distribution. This construction is illustrated in Figure~\ref{fig:smoothing-tree-example}. In practice, we could use any ordered binary tree structure on $\{0, 1, \ldots, T\}$ to define a well-posed recursive representation of $\ft{\bbQ}{0}{T}$, but, as we will see in Sections \ref{subsec:stitching} and \ref{subsec:algorithm}, balanced representations offer better statistical and computational properties.
\begin{figure}
    \begin{center}
    \begin{forest}
    for tree={
        grow=south,
        minimum size=3ex, inner sep=1pt,
        s sep=7mm
            }
    [$\ft{\bbQ}{0}{9}$
        [$\ft{\bbQ}{0}{5}$
            [$\ft{\bbQ}{0}{2}$
                [$\ft{\bbQ}{0}{1}$
                    [$\ft{\bbQ}{0}{0}$]
                    [$\ft{\bbQ}{1}{1}$]
                ]
                [$\ft{\bbQ}{2}{2}$]
            ]
            [$\ft{\bbQ}{3}{5}$
                [$\ft{\bbQ}{3}{4}$
                    [$\ft{\bbQ}{3}{3}$]
                    [$\ft{\bbQ}{4}{4}$]
                ]
                [$\ft{\bbQ}{5}{5}$]
            ]
        ]
        [$\ft{\bbQ}{6}{9}$
            [$\ft{\bbQ}{6}{8}$
                [$\ft{\bbQ}{6}{6}$]
                [$\ft{\bbQ}{7}{8}$
                    [$\ft{\bbQ}{7}{7}$]
                    [$\ft{\bbQ}{8}{8}$]
                ]
            ]
            [$\ft{\bbQ}{9}{9}$]
        ]
    ]
\end{forest}
    \caption{Example of a recursive tree structure for $\ft{\bbQ}{0}{9}$.}
    \label{fig:smoothing-tree-example}
    \end{center}
\end{figure}
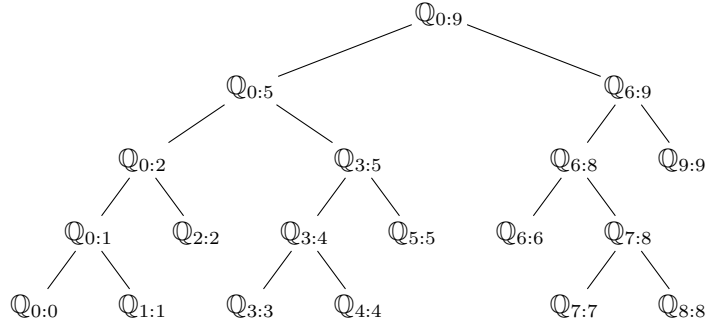

In practice, the stitching operation described by
Proposition~\ref{prop:stitching} is not tractable in closed-form, and we need to
resort to \MC{} integration instead.

\subsection{Sample stitching}\label{subsec:stitching}
    For notational simplicity, in this section and all subsequent ones we do not emphasize the dependency of our estimates on $\nu$. Suppose that we have two independent \MC{} approximations 
\begin{align}
    \ft{\bbQ}{a}{c-1} 
        &\approx \ft{\bbQ^N}{a}{c-1} \coloneqq \sum_{n=1}^N w_{c-1}^n \delta_{\ft{X^n}{a}{c-1}} \\
    \ft{\bbQ}{c}{b} 
        &\approx \ft{\bbQ^N}{c}{b} \coloneqq \sum_{n=1}^N w_{c}^n \delta_{\ft{X^n}{c}{b}},
\end{align} 
following \citet{Lindsten2017Divideandconquer,Kuntz2021productform}, we can then form the ``product-form'' importance empirical density
\begin{align}
    \ft{\widetilde{\bbQ}^N}{a}{b} 
        &\coloneqq \sum_{m, n=1}^N W_c^{mn}\delta_{[\ft{X^m}{a}{c-1}, \ft{X^n}{c}{b}]},\label{eq:estimator}
\end{align}
where 
\begin{align}
    W_c^{mn} = \frac{w_{c-1}^m w_{c}^n \omega_c(X^m_{c-1}, X^n_{c})}{\sum_{i,j=1}^N w_{c-1}^i w_{c}^j\omega_c(X^i_{c-1}, X^j_{c})}.\label{eq:normalized-weights}
\end{align}
As described in \citet{Kuntz2021productform}, this estimator exhibits better properties than the ``naive'' estimator  
\[\sum_{n=1}^N  \frac{w_{c-1}^n w_{c}^n\omega_c(X^n_{c-1},
        X^n_{c-1})}{\sum_{m=1}^N w_{c-1}^m w_{c}^m \omega_c(X^m_{c-1},
X^m_{c-1})}\delta_{[\ft{X^n}{a}{c-1}, \ft{X^n}{c}{b}]}
\]
when the $(X^n_{c-1})_{n=1}^N$, and $(X^n_{c})_{n=1}^N$ have been sampled independently.

Moreover, the denominator of~\eqref{eq:normalized-weights} directly provides us with an estimate of the normalizing constant increment, that is, if $\ft{L^N}{a}{c-1}$ and $\ft{L^N}{c}{b}$ are estimates of the normalizing constants of $\ft{\bbQ}{a}{c-1}$ and $\ft{\bbQ}{c}{b}$, respectively, then, following Proposition~\ref{prop:stitching}, we know that 
\begin{align}
    \ft{L}{a}{b} 
        &= \ft{L}{a}{c-1} \ft{L}{c}{b} \iint \omega_c(x_{c-1}, x_c) \ft{\bbQ}{a}{c-1}(\dd{x_{c-1}}) \ft{\bbQ}{c}{b}(\dd{x_{c}}) \label{eq:Lab}\\
        &\approx \ft{L^N}{a}{c-1} \ft{L^N}{c}{b} \sum_{i,j=1}^N w_{c-1}^i w_{c}^j\omega_c(X^i_{c-1}, X^j_{c-1}).
\end{align}

\begin{algorithm}[t]
\tcp{All operations on indices $m,n$ are done in parallel}
\SetAlgoLined
\Fn{\textsc{Combine}$\left(\ft{X^{1:N}}{a}{c-1}, w^{1:N}_{c-1}, \ft{L^N}{a}{c-1}, \ft{X^{1:N}}{c}{b}, w^{1:N}_{c}, \ft{L^N}{c}{b}\right)$}
{
\tcp{In all but the initial step, the weights $w^m_{c-1}$ and $w^n_{c}$ will be $1/N$.}
$L_c^N \leftarrow \sum_{m,n=1}^N \omega_{c}(X_{c-1}^m, X_{c}^n) w^m_{c-1} w^n_c$\;
$W_c^{mn} \leftarrow \omega_{c}(X_{c-1}^m, X_{c}^n) w^m_{c-1} w^n_c / L_c^N$\;
\tcp{in parallel, using Multinomial, Systematic, etc}
Sample $N$ times from $\sum_{m,n}W_{mn} \delta_{X_{c}^m}(\dd{x_{c}}) \delta_{X_{c-1}^n}(\dd{x_{c-1}})$ to get $(X_{c-1}^{l_n}, X_{c}^{r_n})_{1 \leq n \leq N}$\; 
\tcp{The two loops below can be done in parallel}
\ParFor{$c=a, \ldots, c-1$}{
$\widetilde{X}^{n}_{c'} \leftarrow X_{c'}^{l_n}$ \;
}
\ParFor{$c'=c, \ldots, b$}{
$\widetilde{X}^{n}_{c'} \leftarrow X_{c'}^{r_n}$ \;

}
\Return{$\ft{\widetilde{X}^{1:N}}{a}{b}$, $\ft{L^N}{a}{c-1} \ft{L^N}{c}{b} L_c^N$}
}
\caption{Block combination}
\label{algo:block-combination}
\end{algorithm}

When the importance estimator in~\eqref{eq:estimator} has been formed,
we can then resample $N$ pairs of partial smoothing paths $(l_n, r_n)_{n=1}^N$
according to the normalized weights $W_c^{mn}$ to obtain a stitched \MC{} approximation $\ft{\bbQ^N}{a}{b} \coloneqq \frac{1}{N} \sum_{n=1}^N \delta_{[\ft{X^{l_n}}{a}{c-1}, \ft{X^{r_n}}{c}{b}]}$. This construction is summarized in Algorithm~\ref{algo:block-combination}, from which we can also compute the normalizing constant increment as a by-product.

Under this construction, we can show that the resulting \MC{} $\mathcal{L}_p$ error is well-behaved. The proof of the following proposition is given in Appendix~\ref{app:proof}.

\begin{proposition}\label{prop:consistency}
     Let $p\geq1$ be an integer, suppose that the $(l_n, r_n)_{n=1}^N$ are
     sampled according to a multinomial distribution, that $\varphi$ is a
     bounded measurable function, and that $\omega_c$ is bounded.
     If for any measurable bounded functions $\ft{\varphi}{a}{c-1}$ and $\ft{\varphi}{c}{b}$ we have 
     \begin{align}
         \bbE{}\left[\left|\ft{\bbQ{}}{a}{c-1}(\ft{\varphi}{a}{c-1}) - \ft{\bbQ{}^{N}}{a}{c-1}(\ft{\varphi}{a}{c-1})\right|^p\right]^{1/p} &\leq \ft{C^p}{a}{c-1}\frac{\norm{\ft{\varphi}{a}{c-1}}_{\infty}}{N^{1/2}}\\
         \bbE{}\left[\left|\ft{\bbQ{}}{c}{b}(\ft{\varphi}{c}{b}) - \ft{\bbQ{}^{N}}{c}{b}(\ft{\varphi}{c}{b})\right|^p\right]^{1/p} &\leq \ft{C^p}{c}{b}\frac{\norm{\ft{\varphi}{c}{b}}_{\infty}}{N^{1/2}},
     \end{align}
     for some constants $\ft{C^p}{a}{c-1}, \ft{C^p}{c}{b}$ independent of $N$, $\ft{\varphi}{a}{c-1}$, and $\ft{\varphi}{c}{b}$, then,
     \begin{align}
       \bbE{}\left[\left|\ft{\bbQ{}}{a}{b}(\varphi) - \ft{\bbQ{}^{N}}{a}{b}(\varphi)\right|^p\right]^{1/p} \leq \left(4 \min(\ft{C^p}{a}{c-1}, \ft{C^p}{c}{b})\norm{\bar{\omega}_c}_{\infty} + 2^{(p+1)/p}\right)\frac{\norm{\varphi}_{\infty}}{N^{1/2}}
     \end{align}
     where $\bar{\omega}_c\coloneqq \omega_c/Z_c$, and $Z_c\coloneqq L_{a:b} / L_{a:c-1}
     L_{c:b}$. 
\end{proposition}

Proposition~\ref{prop:consistency} allows us to derive upper bounds to the total $\mathcal{L}_p$ error as a consequence.

\begin{corollary}\label{cor:convergence}
If the $\omega_c/Z_c$'s are uniformly bounded by some constant $\Omega$ independent of $c$, then 
\begin{align}
     \ft{C^p}{0}{T} = O\left((4\Omega)^D\right),
\end{align}
where $D$ is the depth of the tree structure chosen for the smoothing operation (see Figure~\ref{fig:smoothing-tree-example}). In particular, if the tree is balanced, i.e., $D = \ceil{\log_2(T)}$, 
\begin{align}
     \ft{C^p}{0}{T} = O\left(T^{2 + \log_2(\Omega)}\right)
\end{align}
\end{corollary}
\begin{proof}
    The initial case comes from \cite[][Lemma 7.3.3]{DelMoral2004Book} so that for all $c$, $\ft{C^p}{c}{c} \leq 2^{(p+1)/p}$. The result then proceeds by using inequalities on the progression of arithmetico-geometric sequences~\citep[see, e.g.,][Section 4.2.3]{riley2006mathematical}.        
\end{proof}
\begin{remark}
    The uniform bounding of the quantities 
    \begin{align*}
        \frac{\omega_c}{Z_c} = \frac{\omega_c}{\iint \omega_c(x_{c-1}, x_c) \ft{\bbQ}{a}{c-1}(\dd{x_{c-1}}) \ft{\bbQ}{c}{b}(\dd{x_{c}})},
    \end{align*}
    for all $c$, assumed in Corollary~\ref{cor:convergence}, is, for example, true as soon as the $\omega_c$'s are uniformly bounded below and above. This hypothesis, albeit strong, is typically assumed in proofs of the uniform convergence particle filtering algorithms~\citep{del2001stability}.
\end{remark}

Following \citet[][Proof of Lemma 5]{Crisan2000convergence}, for $p>2$, Chebyshev's inequality and Borel--Cantelli lemma also provides the following corollary.
\begin{corollary}
    Under the same hypotheses, $\ft{\bbQ^N}{0}{T}(\varphi)$ converges almost surely to $\ft{\bbQ}{0}{T}(\varphi)$.
\end{corollary}

An interesting point to notice is that for very unbalanced trees with depth of order $T$, Proposition~\ref{prop:consistency} recovers the usual exponential scaling in $T$~\citep{Andrieu2001} of the mean squared error, instead of the polynomial scaling obtained when the tree is balanced. 

\subsection{Algorithm}\label{subsec:algorithm}
    The \MC{} approximation of $\ft{\bbQ}{0}{T}$ can be computed using a recursive algorithm, which can be parallelized across all operations happening at each level of the tree depth. At initialization, we simply need to compute the weights of the particles sampled at $t=0$ as per Proposition~\ref{prop:reweighting}, resulting in Algorithm~\ref{algo:smoother-init}. In order to obtain a balanced tree, we can recursively split at the midpoint of the partial smoothing interval, essentially recovering a binary tree when $T+1$ is a power or $2$. This results in Algorithm~\ref{algo:recursion}.

\begin{algorithm}[t]
\tcp{All operation on indices $m, n$ are done in parallel}
\SetAlgoLined
\ParFor{$t=0, \ldots, T$}{
    $X^n_t \leftarrow \textrm{sample from } q_t(\dd{x_t})$\;
    \If{$t = 0$}{
    $\omega^n_t \leftarrow G_0(X^n_t) \frac{\bbM_0(\dd{x_t})}{q_t(\dd{x_t})}(X^n_t)$\;
        }
    \Else{
        $\omega^n_t \leftarrow \frac{\nu_t(\dd{x_t})}{q_t(\dd{x_t})}(X^n_t)$\;
    }
    $W_t^n \leftarrow w^n_t / \sum_{m=1}^N w_t^m$\;
    $L_t^N \leftarrow \frac{1}{N}\sum_{m=1}^N w_t^m$\;
}
\Return{$\ft{X^{1:N}}{0}{T}, \ft{W^{1:N}}{0}{T}, \ft{L^N}{0}{T}$}
\caption{Smoother initialization}
\label{algo:smoother-init}
\end{algorithm}

\begin{algorithm}[t]
\SetAlgoLined
\Fn{\textsc{Recursion}$\left(\ft{X^{1:N}}{a}{c-1}, \ft{w^{1:N}}{a}{c-1}, \ft{L^N}{a}{c-1}, \ft{X^{1:N}}{c}{b}, \ft{w^{1:N}}{c}{b}, \ft{L^N}{c}{b}\right)$}{
    \If{$a=c-1$ and $b=c$}
        {
        \Return{\textsc{Combine}$\left(\ft{X^{1:N}}{a}{a}, w^{1:N}_a, \ft{L^N}{a}{a}, \ft{X^{1:N}}{b}{b}, w^{1:N}_b, \ft{L^N}{b}{b}\right)$}
        }
    \ElseIf{$c-1 > a$ and $b=a$}
        { 
        $c' \leftarrow \floor{\frac{a + c-1}{2}}$\;
        $\ft{X^{1:N}}{a}{c-1}, \ft{L^N}{a}{c-1} \leftarrow$ \textsc{Recursion}$\left(\ft{X^{1:N}}{a}{c'-1}, \ft{w^{1:N}}{a}{c'-1}, \ft{L^N}{a}{c'-1}, \ft{X^{1:N}}{c'}{c-1}, \ft{w^{1:N}}{c'}{c-1}, \ft{L^N}{c'}{c-1}\right)$\;
        \Return{\textsc{Combine}$\left(\ft{X^{1:N}}{a}{c-1}, (1 / N)_{n=1}^N, \ft{L^N}{a}{c-1}, \ft{X^{1:N}}{b}{b}, w^{1:N}_b, \ft{L^N}{b}{b}\right)$}
        }
    \ElseIf{$a=c-1$ and $b > c$}
        { 
        $c' \leftarrow \floor{\frac{c + b}{2}}$\;
        $\ft{X^{1:N}}{c}{b}, \ft{L^N}{c}{b} \leftarrow$ \textsc{Recursion}$\left(\ft{X^{1:N}}{c}{c'-1}, \ft{w^{1:N}}{c}{c'-1}, \ft{L^N}{c}{c'}, \ft{X^{1:N}}{c'}{b}, \ft{w^{1:N}}{c'}{b}, \ft{L^N}{c'}{b}\right)$\;
        \Return{\textsc{Combine}$\left(\ft{X^{1:N}}{a}{a}, w_a^{1:N}, \ft{L^N}{a}{a}, \ft{X^{1:N}}{c}{b}, (1 / N)_{n=1}^N, \ft{L^N}{c}{b}\right)$}
        }
    \Else
        { 
        $c' \leftarrow \floor{\frac{a + c-1}{2}}$\;
        $\ft{X^{1:N}}{a}{c-1}, \ft{L^N}{a}{c-1} \leftarrow$ \textsc{Recursion}$\left(\ft{X^{1:N}}{a}{c'-1}, \ft{w^{1:N}}{a}{c'-1}, \ft{L^N}{a}{c'-1}, \ft{X^{1:N}}{c'}{c-1}, \ft{w^{1:N}}{c'}{c-1}, \ft{L^N}{c'}{c-1}\right)$\;
        $c' \leftarrow \floor{\frac{c + b}{2}}$\;
        $\ft{X^{1:N}}{c}{b}, \ft{L^N}{c}{b} \leftarrow$ \textsc{Recursion}$\left(\ft{X^{1:N}}{c}{c'-1}, \ft{w^{1:N}}{c}{c'-1}, \ft{L^N}{c}{c'}, \ft{X^{1:N}}{c'}{b}, \ft{w^{1:N}}{c'}{b}, \ft{L^N}{c'}{b}\right)$\;
        \Return{\textsc{Combine}$\left(\ft{X^{1:N}}{a}{c-1}, (1 / N)_{n=1}^N, \ft{L^N}{a}{c-1}, \ft{X^{1:N}}{c}{b}, (1 / N)_{n=1}^N, \ft{L^N}{c}{b}\right)$}
        }
    }
\caption{Recursion}
\label{algo:recursion}
\end{algorithm}

The smoothing algorithm then simply consists in passing the output of Algorithm~\ref{algo:smoother-init} to Algorithm~\ref{algo:recursion}.
It is worth noting that while Algorithm~\ref{algo:recursion} is correct, its recursive nature makes its implementation on parallel devices tedious if one wants to benefit from hardware acceleration. Moreover it does not consist in a tail recursion~\cite[see, e.g.][Ch. 15]{muchnick1997advanced}, so that it cannot easily be transformed into a loop that would be easier to parallelize. However, the split-combine operations can be reformulated as a series of tensor reshaping operations, which is more amenable to parallelization. We provide this equivalent, albeit parallelizable, formulation of the algorithm in Appendix~\ref{app:parallel-combination}. 

Consider now the choice of tree partitioning given in Figure~\ref{fig:smoothing-tree-example}. The nodes correspond to the combination operation, while the edges correspond to the split happening in Algorithm~\ref{algo:recursion}. All the operations at a given depth can be run fully in parallel, each of them being entirely parallelizable too with respect to the particle samples, except for the resampling operation. The resampling operation requires normalizing the weights and running parallel search operations, which can be done with span complexity \citep{Cormen:2009} of $\bigO(\log_2(N))$ on parallel architectures~\citep[using prefix-sum operations, see, e.g.]{Murray2016ParallelResampling}, so that each level of the tree has span complexity $\bigO(\log_2(N))$. This results in a parallelized algorithm run time that globally scales linearly with the depth of the smoothing tree considered, and logarithmically in the number of particles. As a consequence, we have the following proposition.

\begin{proposition}\label{prop:complexity}
     The total span complexity of dSMC is $\bigO(\log_2(T) \log_2(N))$. 
\end{proposition}

\begin{remark}
    It is worth noting that some alternative resampling methods exist that allow to parallelize the resampling operation, at the cost of biasing it, or at the cost of random execution time~\citep{Murray2016ParallelResampling}. We discuss these methods and the additional benefits they provide for dSMC in Section~\ref{subsec:parallel-resampling}.
\end{remark}

\section{Parallel-in-time particle Gibbs}
\label{sec:pgibbs}
We now focus on deriving a conditional formulation of dSMC (that we call c-dSMC) that we then use to build a PIT particle Gibbs algorithm. We quickly discuss its degeneracy properties, and in particular the fact that it does not require a backward sampling step to help with mixing the initial time steps.

\subsection{Conditional dSMC sampler}
    \label{subsec:conditional-smc}
    Particle Gibbs methods were introduced in \citet{Andrieu2010particle} in order to  sample from the joint posterior $\ft{\bbQ}{0}{T}(\dd{\ft{x}{0}{T}}, \theta)$ of a  state-space model. It consists in successively applying two conditional sampling steps: (i) sampling $\theta'$ conditionally on a given smoothing trajectory $\ft{x^*}{0}{T}$, and (ii) sampling a smoothing trajectory $\ft{x'}{0}{T}$ conditionally on $\ft{x^*}{0}{T}$. Step (ii) needs to be understood as ``conditionally to one of the trajectories being sampled by the SMC algorithm being $\ft{x^*}{0}{T}$''. 

Due to the arbitrary tree representation of the smoothing operation, it is complicated to manipulate the complete expression for the distribution of all the random variables generated during the course of the algorithm\footnote{Although this was used for instance in \cite{Lindsten2017Divideandconquer} to prove the unbiasedness of their resulting likelihood estimate.}. However, we can still easily provide a natural recursive expression that will serve as a support for understanding the behavior of the conditional distributions. In order to make notations simpler, we write $\ft{\sigma}{a}{c-1}(k)$, $k = a, a+1, \ldots c-1$ for the resampling array applied to node $k$, and we write $\ft{\sigma^n}{a}{c-1}(k)$ for its $n$-th element (and similarly for $\ft{\sigma}{c}{b}(k)$).
    
\begin{remark}
    $\ft{\sigma}{a}{c-1}$ is a function of the left-right resampling indices $\ft{l^{1:N}}{a+1}{c-1}, \ft{r^{1:N}}{a+1}{c-1}$ generated deeper in the recursion tree, and similarly for $\ft{\sigma}{c}{b}$ via the recursion 
    \begin{align*}
        \ft{\sigma^m}{a}{b}(k) = \ft{\sigma^{l_c^m}}{a}{c-1}(k), \quad \text{for all $a \leq k < c$}\\
        \ft{\sigma^m}{a}{b}(k) = \ft{\sigma^{r_c^m}}{c}{b}(k), \quad \text{for all $c \leq k \leq b$},
    \end{align*}
    for all $m \in \{1, 2, \ldots, N\}$.
\end{remark}
    
For $a<c<b$, let $\ft{\psi}{a}{c-1}(\dd{\ft{x^{1:N}}{a}{c-1}}, \ft{l^{1:N}}{a+1}{c-1}, \ft{r^{1:N}}{a+1}{c-1})$ and $\ft{\psi}{c}{b}(\dd{\ft{x^{1:N}}{c}{b}}, \ft{l^{1:N}}{c+1}{b}, \ft{r^{1:N}}{c+1}{b})$ be the full distributions of all the random variables generated by dSMC for the partial smoothing distributions $\ft{\bbQ}{a}{c-1}$ and $\ft{\bbQ}{c}{b}$, respectively.
The full distribution of all the random variables generated to sample from $\ft{\bbQ}{a}{b}$ is given by
\begin{align}
    \ft{\psi}{a}{b}(\ft{\dd{x}^{1:N}}{a}{b}, \ft{l^{1:N}}{a+1}{b}, \ft{r^{1:N}}{a+1}{b}) 
        &=\ft{\psi}{a}{c-1}(\ft{\dd{x}^{1:N}}{a}{c-1}, \ft{l^{1:N}}{a+1}{c-1}, \ft{r^{1:N}}{a+1}{c-1}) \\
        &\quad\times \ft{\psi}{c}{b}(\ft{\dd{x}^{1:N}}{c}{b}, \ft{l^{1:N}}{c+1}{b}, \ft{r^{1:N}}{c+1}{b}) \\
        &\quad\times \left\{\prod_{n=1}^N W_c^{l_c^n, r_c^n}\right\}.
\end{align}
where for all $m, n \in \{1, \ldots, N\}$, we have
\begin{align}
    W_c^{m,n} \propto \omega_c\left(x_{c-1}^{\ft{\sigma^m}{a}{c-1}(c-1)}, x_c^{\ft{\sigma^n}{c}{b}(c)}\right)
\end{align}
so that $\sum\limits_{m,n=1}^N W_c^{m,n} = 1$.
This also provides the marginal distribution (note how $l_c$ and $r_c$ are not present in the arguments of $\psi$)
\begin{align}
    \ft{\psi}{a}{b}(\ft{\dd{x}^{1:N}}{a}{b}, \ft{l^{1:N}}{a+1}{c-1}, \ft{l^{1:N}}{c+1}{b}, \ft{r^{1:N}}{a+1}{c-1}, \ft{r^{1:N}}{c+1}{b}) 
        &=      \ft{\psi}{a}{c-1}(\ft{\dd{x}^{1:N}}{a}{c-1}, \ft{l^{1:N}}{a+1}{c-1}, \ft{r^{1:N}}{a+1}{c-1}) \\
        &\quad\times \ft{\psi}{c}{b}(\ft{\dd{x}^{1:N}}{c}{b}, \ft{l^{1:N}}{c+1}{b}, \ft{r^{1:N}}{c+1}{b}).
\end{align}
Similarly, the related estimate of the normalizing constant 
\begin{align}
    \ft{L}{a}{b}^N = \ft{L}{a}{b}^N(\ft{x}{a}{b}^{1:N}, \ft{l^{1:N}}{a+1}{c-1}, \ft{l^{1:N}}{c+1}{b}, \ft{r^{1:N}}{a+1}{c-1}, \ft{r^{1:N}}{c+1}{b})
\end{align}
follows the recursion
\begin{align}
    \ft{L}{a}{b}^N = \ft{L}{a}{c-1}^N \ft{L}{c}{b}^N \left\{\frac{1}{N^2} \sum_{i,j=1}^N \omega_c\left(x_{c-1}^{\ft{\sigma^i}{a}{c-1}(c-1)}, x_c^{\ft{\sigma^j}{c}{b}(c)}\right)\right\}\label{eq:rec-loglik}.
\end{align}
Putting these together allows us to characterize recursively the invariant distribution of the PMMH, that we will then use in order to express the related conditional dSMC distribution
\begin{align}
    \ft{\pi}{a}{b}(\ft{\dd{x}^{1:N}}{a}{b}, \ft{l^{1:N}}{a+1}{c-1}, \ft{l^{1:N}}{c+1}{b}, \ft{r^{1:N}}{a+1}{c-1}, \ft{r^{1:N}}{c+1}{b}) 
        = \frac{\ft{L}{a}{b}^N}{\ft{L}{a}{b}} \, \ft{\psi}{a}{b}(\ft{\dd{x}^{1:N}}{a}{b}, \ft{l^{1:N}}{a+1}{c-1}, \ft{l^{1:N}}{c+1}{b}, \ft{r^{1:N}}{a+1}{c-1}, \ft{r^{1:N}}{c+1}{b}).
\end{align}
Under these notations, we can define the star trajectory recursively as $\ft{\sigma^*}{a}{c-1} = \ft{\sigma^I}{a}{c-1}$, $\ft{\sigma^*}{c}{b} = \ft{\sigma^J}{c}{b}$, where $(I,J)$ is distributed according to a multinomial distribution on $W^{m,n}_c$, and $\ft{\sigma^*}{a}{b} = \left[\ft{\sigma^*}{a}{c-1}, \ft{\sigma^*}{c}{b}\right]$. This in turn defines $\ft{X^*}{a}{b} \coloneqq \ft{X^{\ft{\sigma^*}{a}{b}}}{a}{b} = \left[\ft{X^{*}}{a}{c-1}, \ft{X^*}{c}{b}\right]$, and the related $l_c^*, r_c^*$ which correspond to the resampling indices pairs that eventually lead to the star trajectory.
\begin{proposition}[Conditional dSMC]
    $\ft{X^*}{a}{b}$ is marginally distributed according to $\ft{\bbQ}{a}{b}$.
\end{proposition}
\begin{proof}
    In order to show this, we only need to prove that the star trajectory is marginally distributed according to the smoothing distribution $\ft{\pi}{a}{b}$. Thanks to the recursive definition of $\ft{L^N}{a}{b}$ and $\ft{\psi}{a}{b}$, we can see that the normalizing term $\sum_{i,j=1}^N \omega_c\left(x_{c-1}^{\ft{\sigma^i}{a}{c-1}(c-1)}, x_c^{\ft{\sigma^j}{c}{b}(c)}\right)$ disappears in the definition of $\ft{\pi}{a}{b}$. This ensures that the resulting distribution involves only products of $\omega_c$ evaluated at particles belonging to other trajectories than the star-trajectory. In particular, marginalizing trajectories other than the star-trajectory consists in isolating its factors in $\ft{\pi}{a}{b}$. This results in
    \begin{align}
        \ft{\pi}{a}{b}(\ft{\dd{x}^*}{a}{b}) \propto \prod_{c=a}^b q_c(\dd{x_c^*}) \prod_{c=a}^b \omega_c(x_{c-1}^*, x_c^*)
    \end{align}
    which is exactly the target smoothing distribution. 
\end{proof}
\begin{remark}
    This construction also generalizes to the algorithm of \citet{Lindsten2017Divideandconquer} provided that one is able to sample conditional distributions at each node of their construction, for example using the sequential method of \citet{Lindsten2014AncestorSampling}.
\end{remark}
\subsection{Parallel-in-time particle Gibbs}
    \label{subsec:pgibbs}
    The resulting algorithm resembles the classical conditional SMC algorithm of \cite{Andrieu2010particle}, in that, similarly, we can implement it by simply enforcing that the first trajectory be preserved throughout the course of the recursion. In particular, only Algorithms \ref{algo:block-combination} and \ref{algo:smoother-init} need to be modified. The conditional version of Algorithm \ref{algo:smoother-init} simply consists in prepending the star trajectory to the sampled proposal trajectories before computing the resulting weights. On the other hand the conditional version of Algorithm~\ref{algo:block-combination} consists in preserving said star trajectory throughout the resampling steps and is given by Algorithm~\ref{algo:conditional-block-combination}.

\begin{algorithm}
\tcp{All operations on indices $m, n$ are done in parallel}
\SetAlgoLined
\KwResult{Combine conditional particle representation of partial smoothing distributions}
\Fn{\textsc{ConditionalCombine}$\left(\ft{X^{1:N}}{a}{c-1}, w^{1:N}_{c-1}, \ft{X^{1:N}}{c}{b}, w^{1:N}_{c}\right)$}
{
$W_{mn} \leftarrow \omega_{c}(X_{c-1}^m, X_{c}^n) w^m_{c-1} w^n_c$\;
Sample independently $N-1$ times from $\sum_{m,n}W_{mn} \delta_{X_{c}^m}(\dd{x_{c}}) \delta_{X_{c-1}^n}(\dd{x_{c-1}})$ to get $(X_{c-1}^{l_n}, X_{c}^{r_n})_{1 \leq n \leq N}$\ \tcp{in parallel, using Multinomial}\;
\tcp{The two loops below can be done in parallel}
\ParFor{$c'=a, \ldots, c-1$}{
    $\widetilde{X}^{n}_{c'} \leftarrow X_{c'}^{1}$ \;
    $\widetilde{X}^{n}_{c'} \leftarrow X_{c'}^{l_n}$ \;
}
\ParFor{$c'=c, \ldots, b$}{
    $\widetilde{X}^{n}_{c'} \leftarrow X_{c'}^{1}$ \;
    $\widetilde{X}^{n}_{c'} \leftarrow X_{c'}^{r_n}$ \;
}
\Return{$\ft{\widetilde{X}^{1:N}}{a}{b}$}
}
\caption{Conditional Block combination}
\label{algo:conditional-block-combination}
\end{algorithm}

\cite{Andrieu2010particle}  considered implementing the conditional SMC
step using a particle filter only, which resulted in lower mixing speeds for
time steps further away from the last time step $T$. This was corrected by the
introduction of the so-called backward sampling step~\citep{Whiteley_disc_PMCMC, Lindsten2012}, which
enabled rejuvenating the conditional trajectories; see also \cite{Lindsten2014AncestorSampling} for a related approach. A noteworthy point is that
our proposed PIT particle Gibbs algorithm does not suffer from the classical
genealogy degeneracy problem that prompted the development of the ancestor sampling
step. This is due to the fact that the degeneracy arising in
dSMC is essentially uniform across all time steps thanks to the balanced tree
structure. Indeed, instead of the last time steps being resampled just a few
times and the initial time steps being resampled around $T$ times, as in
standard SMC, all time steps in dSMC are resampled at most $\ceil{\log_2(T)}$
times. This is also the reason why the $\mathcal{L}_p$ error in
Proposition~\ref{prop:consistency} scales as a polynomial of $T$ and not
exponentially. In practice, this means that the modified trajectories sampled
from our conditional dSMC will mix similarly for initial timesteps and for
final ones, provided that our proposal distributions $q_t$ and auxiliary weight
functions $\nu_t$ are adapted to the model and data at hand.

\section{Variance reduction methods}
\label{sec:proposals}
A drawback of our method consists in the necessity to use independent proposals $q_{0:T}(\dd{x_{0:T}}) = \prod_{t=0}^T q_{t}(\dd{x_{t}})$. It is well known that using such rough estimates increases the variance of the smoothing distribution estimates in particular in case of ``sticky'' processes which exhibit a strong time-dependency, or more precisely, when the conditional reverse Markov chain representing the smoothing distribution mixes slowly. However, this problem can be mitigated by using proposal distributions that are adapted to the model at hand. In Section~\ref{subsec:pit-gaussian} we describe how recently developed parallel-in-time Gaussian approximation based smoothing algorithms~\citep{Sarkka2021ParallelKF,Yaghoobi2021ParallelEKF} can be used to form such proposal. As these methods are also parallel in time, they do not relinquish the $\bigO(\log(T))$ span complexity of the dSMC algorithm. 

More prosaically, a natural way to reduce the variance of the smoothing estimators is to increase the number of particles used in the \MC{} representations. However, doing so in Algorithm~\ref{subsec:algorithm} comes at a quadratic cost in memory and threads utilization. In Section~\ref{subsec:parallel-resampling} we discuss how we can leverage ideas from \citet{Murray2016ParallelResampling} to lazily resample so as to keep a linear memory cost and reduce the computational burden.

\subsection{Parallel-in-time Gaussian approximated smoothing solutions}
    \label{subsec:pit-gaussian}
    It is well known that non-linear SSMs for which the state posterior distribution is unimodal can be approximated by LGSSMs. For example, consider an additive Gaussian noise transition model $p_t(x_t \mid x_{t-1}) = \mathcal{N}(x_t; f(x_{t-1}), Q_{t-1}) \, \dd{x_t}$. Under the Gaussian approximated assumption $p(x_t \mid y_{1:t}) \approx \mathcal{N}(x_t; m_t, P_t)$, we can use a Taylor linearization of the transition function $f$ around the approximated mean $m_t$ to form the linearized dynamics $x_{t+1} = f(m_t) + J[f](m_t)(x_t -m_t) + \epsilon_{t}$, where  $\epsilon_{t}$ is a Gaussian random variable with mean $0$ and covariance $J[f](m_t) Q_{t} J[f](m_t)^\top$ and $J[f](m_t)$ is the Jacobian of $f$ evaluated at $m_t$. By repeating this approximation for each time step and for the observation model, we obtain the extended Kalman filter algorithm \citep{Jazwinski:1970}. Similarly, one can use Taylor expansion in order to compute Gaussian approximations of the smoothing distribution marginals $p(x_t \mid y_{1:T})$ for all $t$, yielding the extended Kalman smoother algorithm. Other linearization techniques exist, such as statistical linearization \citep{Gelb1974appliedestimation}, sigma-point (unscented) methods \citep{Julier+Uhlmann+Durrant-Whyte:2000,Sarkka:2008}, and numerical integration based methods~\citep{Ito+Xiong:2000,Sarkka+Hartikainen:2010a}. For a review, we refer the reader to \citet{Sarkka2013Book}.

In practice it is worth noting that the reference point used to linearize the system at time $t$ ($m_t$ for the extended Kalman filter example above) is arbitrary, and could be optimized instead of taking the result of the previous time step. This remark led to development of iterated extended Kalman filters \citep{Bell:1993}, iterated sigma-point filters \citep{Sibley2006iteratedukf,zhan2007iterated}, and general iterated statistical linear regression methods called posterior linearization filters \citep{Garcia:2016}. When considering smoothing problems, it is even better to iteratively linearize with respect to the smoothing trajectory as is done in the iterated extended Kalman smoother \citep{Bell:1994}. A general framework of iterated posterior linearization smoothers using this idea was developed in \citet{Garcia:2017} and this was further generalized to more general state-space models in \citet{Tronarp:2018}. These methods result in Gaussian approximations to the marginals $p(x_t \mid y_{1:T}) \approx \mathcal{N}(x_t; m_t^l, P_t^l)$ which are optimal in a Kullback--Leibler sense \citep{Garcia:2016}.

Recently, \citet{Sarkka2021ParallelKF} showed that by reformulating Bayesian filters and smoothers (including Kalman filters and smoothers) in terms of associative operators, it is possible to parallelize them along the time dimension by leveraging prefix-sum algorithms \citep{Blelloch:1989}. This leads to logarithmic span-time complexity $\log(T)$ instead of the conventional $\bigO(T)$ of sequential methods. \citet{Yaghoobi2021ParallelEKF} then extended this framework to non-linear models by developing parallelized versions of the iterated extended Kalman smoothers as well as the more general iterated posterior linearization smoothers. This framework allows for computing the marginal approximations $p(x_t \mid y_{1:T}) \approx \mathcal{N}(x_t; m_t^l, P_t^l)$ in the $\bigO(\log T)$ time complexity.

These Gaussian approximations to the smoothing distributions can now be used as proposal distributions $q_t$ and/or weighting distributions $\nu_t$ in the proposed dSMC algorithm. The resulting method with $q_t = \nu_t$ is summarized in Algorithm~\ref{algo:pit-gaussian-proposal}. 

\begin{algorithm}
\SetAlgoLined
\Fn{\textsc{LinearizedSmoother}$\left(y_{1:T}\right)$}
{
\ParFor{$t=0, \ldots, T$}{
    Initialize $q_t^0 = \mathcal{N}(m_t^0, P_t^0)$ \tcp{for example using the stationary distribution} \;
}
\While{criterion not verified}{
    Linearize \eqref{eq:ssm} around $q^l_t, t=0, 1, \ldots, T$ \tcp{Done in parallel} \;
    Run parallel Kalman filter and RTS smoothers on the linearized system as per \cite{Sarkka2021ParallelKF,Yaghoobi2021ParallelEKF}\;
    Set $p(x_t \mid y_{1:T}) \approx q^l_t = \mathcal{N}(m^l_t, P^l_t),  t=0, 1, \ldots, T$\;
}
Run the parallel smoother defined as per Algorithms~\ref{algo:smoother-init} and \ref{algo:block-combination}\;
}
\caption{PIT linearized proposal smoother}
\label{algo:pit-gaussian-proposal}
\end{algorithm}

Similarly, we can tweak Algorithm~\ref{algo:pit-gaussian-proposal} in order to define an efficient Gaussian proposal model for PIT pGibbs. Indeed, between two iterations of the d-cSMC described in Section~\ref{subsec:conditional-smc}, pGibbs typically proposes new parameters. We can expect that the parameters of the state-space model to not have changed too much. Intuitively, this means that the optimum trajectory for the parallel IPLS method will not change much and we can therefore reuse the optimum of the previous Gibbs iteration as initialization for the next one. The benefit of doing so is shown in the experiment of Section~\ref{subsec:exp_gibbs}.

\subsection{Parallel resampling for lazy evaluation of the weight matrix}
    \label{subsec:parallel-resampling}
    Algorithm~\ref{algo:block-combination} presented in Section~\ref{subsec:algorithm} requires to form a $N\times N$ matrix to then sample $N$ elements from it. Doing so limits the scalability of dSMC in at least two ways.

\begin{enumerate}
    \item The memory cost will increase quadratically with the required number of particles. This is particularly problematic on parallel hardware such as GPUs where the memory available is usually more limited than the main (random-access memory) memory accessible via a CPU. For a large number of time steps or particles, our algorithm may therefore simply fail to return a result.
    \item The number of threads available on GPUs, while increasing year-on-year, is still limited, and our algorithm computational scalability, although theoretically logarithmic in both $N$ and $T$, may be affected by threading bottlenecks. See Section~\ref{subsec:ffbs_comp} for an illustration of this. 
\end{enumerate}

In order to mitigate both these issues, we can leverage the parallel resampling schemes proposed by \citet{Murray2016ParallelResampling}. Indeed, these can be modified in order to sample $N$ entries from a set of $N \times N$ unnormalized weights \emph{without needing to evaluate the whole matrix}. This property, although not discussed in \citet{Murray2016ParallelResampling} can crucially be utilized to design lazy resampling schemes for our $N \times N$ size importance density~\eqref{eq:estimator}. 
Formally, suppose we want to sample $N$ pairs $(I_m, J_m)$ independently from a Multinomial distribution $\mathrm{Mul}((W_c^{i,j})_{i, j=1}^N)$, where for all $i, j$, $W_c^{i,j} \propto \omega_{c}(X^i_{c-1}, X^j_{c})$ for some time index $c$. This can be done in parallel across the $N$ pairs $(I_m, J_m)$ by considering $N$ independent instances of a Metropolis-Hastings~\citep[Code 2 in][]{Murray2016ParallelResampling} algorithm with proposal (after proper flattening of the $N \times N$ matrix) $\mathcal{U}(\{1, \ldots, N^2\})$ 
and target $\propto \omega_{c}(X^i_{c-1}, X^j_{c})$. Similarly, when an upper bound $\overline{\omega}_c$ to $\omega_c$ is available, an unbiased rejection sampling equivalent~\citep[Code 3 in][]{Murray2016ParallelResampling} can be implemented. Under this perspective, we only need to evaluate the term $\omega_{c}(X^i_{c-1}, X^j_{c})$ for the proposed pairs $(i, j)$. This allows us to never increase the memory and thread utilization beyond $\bigO(N)$ operations at any point in time.  On this other hand, this also means that we may inefficiently re-evaluate the same pair several times. However, as shown in Section~\ref{subsec:lazy-experiment}, the parallelization makes this trade-off beneficial. For the sake of completeness, we reproduce the resulting resampling algorithms in Appendix~\ref{app:lazy-resampling}.

Finally, while using these lazy resampling schemes comes at a price (biasedness in the case of the Metropolis-Hastings variation and random execution time in the case of the rejection sampling one), as discussed in \citet[Sections 3.2 and 3.4]{Murray2016ParallelResampling}, this trade-off becomes better as the variance of the weights $\omega_{c}(X^i_{c-1}, X^j_{c})$ decreases.

\section{Experiments}
\label{sec:experiments}
In order to illustrate the computational and statistical properties of our proposed methods, we now consider a set of examples from the literature and compare with the sequential counterparts of our methods. All the results were obtained using an Nvidia\textsuperscript{\textregistered} GeForce RTX 3090 GPU with 24GB memory and the code to reproduce them can be found at \url{https://github.com/AdrienCorenflos/parallel-ps}.
\subsection{Comparison with FFBS}
    \label{subsec:ffbs_comp}
    
In this section, we compare dSMC to the classical forward filtering backward sampling (FFBS) algorithm \citep{godsill2004monte}, both in terms of execution time and Monte Carlo error.
To make the comparison fairer, we also implement FFBS on GPU; in this way, FFBS  scale as $\bigO(T \log(N))$~\citep[see, e.g., the prefix-sum implementation of classical resampling operations in][]{Murray2016ParallelResampling}, since  the particle operations are parallelisable up to a logarithmic factor (corresponding to computing the sum of the importance weights, which can be done using a prefix-sum algorithm). 
We consider the same model as in \citet{Chopin2015particle} (which is a simplified version of the model in  \citet{yu2011center} for photon emission):
\begin{align}
    x_0 &\sim \mathcal{N}\left(\mu, \frac{\sigma^2}{1 - \rho^2}\right),\\
    x_t &= \mu + \rho(\lambda x_{t-1} - \mu) + \epsilon_{t-1}, \quad \epsilon_{t-1} \sim \mathcal{N}(0, \sigma^2), \quad t \geq 1,\\
    y_t &\sim \mathcal{P}(\exp(x_t)),
\end{align}
where $\mathcal{P}(\exp(x_t))$ denotes a Poisson distribution with rate $\exp(x_t)$, and we want to estimate its Fisher score with respect to $\sigma^2$, and evaluated at $\sigma^2$:
\begin{align}
    &\bbE\left[\nabla_{\sigma^2} \ln p(X_{0:T}, y_{0:T}) \mid y_{0:T}\right]\\
    &=\bbE\left[-\frac{T+1}{2 \sigma^2} + \frac{1 - \rho^2}{2 \sigma^4}(X_0 - \mu)^2 + \frac{1}{2\sigma^4}\sum_{s=1}^T \left\{X_s - \mu - \rho(X_{s-1} - \mu)\right\}^2 \mid y_{0:T} \right].
\end{align} 
Because of its additive nature, the variance of this expectation should increase as $T$ increases, making it a good benchmark function to test our algorithm.

The stationary distribution of the underlying dynamics is $x_t \sim \mathcal{N}(\mu, \sigma^2/ (1 - \rho^2))$, so we take $q_t = \nu_t = \mathcal{N}(\mu, \sigma^2/ (1 - \rho^2))$ for all $t$.

In order to study the statistical and numerical properties of our algorithm we then generate data $x_{0:T}, y_{0:T}$ from the model for $T=32, 64, 128, 256, 512$ and repeat $100$ dSMC and FFBS smoothing experiments on the same data.

\begin{figure}[!htb]
    \centering
    \resizebox{\linewidth}{!}{
    \begin{tikzpicture}
\pgfplotsset{every tick label/.append style={font=\tiny}}

\pgfplotsset{
    legend image with text/.style={
        legend image code/.code={%
            \node[anchor=center] at (0.3cm,0cm) {#1};
        }
    },
}

\begin{axis}[
    xlabel shift = -3 pt,
    ylabel shift = -3 pt,
    height = 6cm,
    width = 10cm,
    legend style={nodes={scale=0.5, transform shape}},
    legend pos=north west,
    legend columns=2, %
    log basis x={2},
    log basis y={10},
    tick align=outside,
    xlabel={\scriptsize $T$},
    xmode=log,
    ylabel={\scriptsize Runtime (s)},
    ymode=log,
]

\addlegendimage{legend image with text=dSMC}
\addlegendentry{}
\addlegendimage{legend image with text=FFBS}
\addlegendentry{}

\addplot [thick, color0, mark=*, mark size=2, mark options={solid,draw=white}] 
    table [x=T,y expr={ifthenelse(\thisrow{N}==25,\thisrow{dSMC},nan)}, col sep=comma]
    {sections/experiments/data/cox-results.csv};
\addlegendentry{}
\addplot [thick, color0, mark=square*, mark size=2, mark options={solid,draw=white}] 
    table[x=T,y expr={ifthenelse(\thisrow{N}==25,\thisrow{FFBS},nan)}, col sep=comma]
    {sections/experiments/data/cox-results.csv};
\addlegendentry{$N = 25$}  
\addplot [thick, color1, mark=*, mark size=2, mark options={solid,draw=white}] 
    table [x=T,y expr={ifthenelse(\thisrow{N}==250,\thisrow{dSMC},nan)}, col sep=comma]
    {sections/experiments/data/cox-results.csv};
\addlegendentry{}
\addplot [thick, color1, mark=square*, mark size=2, mark options={solid,draw=white}] 
    table[x=T,y expr={ifthenelse(\thisrow{N}==250,\thisrow{FFBS},nan)}, col sep=comma]
    {sections/experiments/data/cox-results.csv};
\addlegendentry{$N = 50$}  
\addplot [thick, color2, mark=*, mark size=2, mark options={solid,draw=white}] 
    table [x=T,y expr={ifthenelse(\thisrow{N}==100,\thisrow{dSMC},nan)}, col sep=comma]
    {sections/experiments/data/cox-results.csv};
\addlegendentry{}
\addplot [thick, color2, mark=square*, mark size=2, mark options={solid,draw=white}] 
    table[x=T,y expr={ifthenelse(\thisrow{N}==100,\thisrow{FFBS},nan)}, col sep=comma]
    {sections/experiments/data/cox-results.csv};
\addlegendentry{$N = 100$}  
\addplot [thick, color3, mark=*, mark size=2, mark options={solid,draw=white}] 
    table [x=T,y expr={ifthenelse(\thisrow{N}==250,\thisrow{dSMC},nan)}, col sep=comma]
    {sections/experiments/data/cox-results.csv};
\addlegendentry{}
\addplot [thick, color3, mark=square*, mark size=2, mark options={solid,draw=white}] 
    table[x=T,y expr={ifthenelse(\thisrow{N}==250,\thisrow{FFBS},nan)}, col sep=comma]
    {sections/experiments/data/cox-results.csv};
\addlegendentry{$N = 250$}  
\addplot [thick, color4, mark=*, mark size=2, mark options={solid,draw=white}] 
    table [x=T,y expr={ifthenelse(\thisrow{N}==500,\thisrow{dSMC},nan)}, col sep=comma]
    {sections/experiments/data/cox-results.csv};
\addlegendentry{}
\addplot [thick, color4, mark=square*, mark size=2, mark options={solid,draw=white}] 
    table[x=T,y expr={ifthenelse(\thisrow{N}==500,\thisrow{FFBS},nan)}, col sep=comma]
    {sections/experiments/data/cox-results.csv};
\addlegendentry{$N = 500$} 
\addplot [thick, color5, mark=*, mark size=2, mark options={solid,draw=white}] 
    table [x=T,y expr={ifthenelse(\thisrow{N}==1000,\thisrow{dSMC},nan)}, col sep=comma]
    {sections/experiments/data/cox-results.csv};
\addlegendentry{}
\addplot [thick, color5, mark=square*, mark size=2, mark options={solid,draw=white}] 
    table[x=T,y expr={ifthenelse(\thisrow{N}==1000,\thisrow{FFBS},nan)}, col sep=comma]
    {sections/experiments/data/cox-results.csv};
\addlegendentry{$N = 1000$}  
\end{axis}

\end{tikzpicture}
    }
    \caption{Average clock time of running a sequential FFBS vs. dSMC. For $T$ small enough, dSMC scales logarithmically, and then linearly when the parallelization threads have all been utilized. The effect is more pronounced for a higher number of particles.} \label{fig:cox-runtime}
\end{figure}
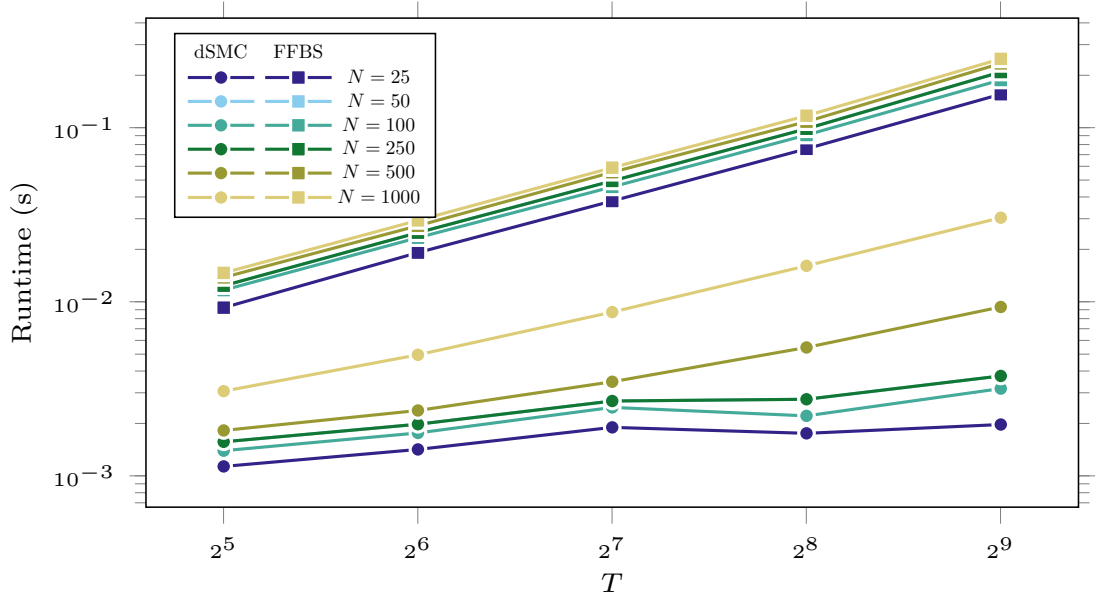

The resulting average running times of the corresponding algorithms are shown in Figure~\ref{fig:cox-runtime}. Our algorithm is always faster that its sequential FFBS counterpart. Due to the limited number of threads on our GPU, the logarithmic complexity scaling of our proposed method reaches a technical upper bound as we increase the number of sampled time steps. In particular the number of time steps that can effectively be parallelized is a decreasing function of the number of particles used. After the parallelization limit has been reached, dSMC scales linearly as further progress is blocked by waiting that a thread becomes free to use.

On the other hand, as can be expected from using independent proposals, our algorithm exhibits a larger error for estimating the Fisher score function, and this error increases with the number of time steps we want to sample. This effect is illustrated by Figure~\ref{fig:cox-error}. 
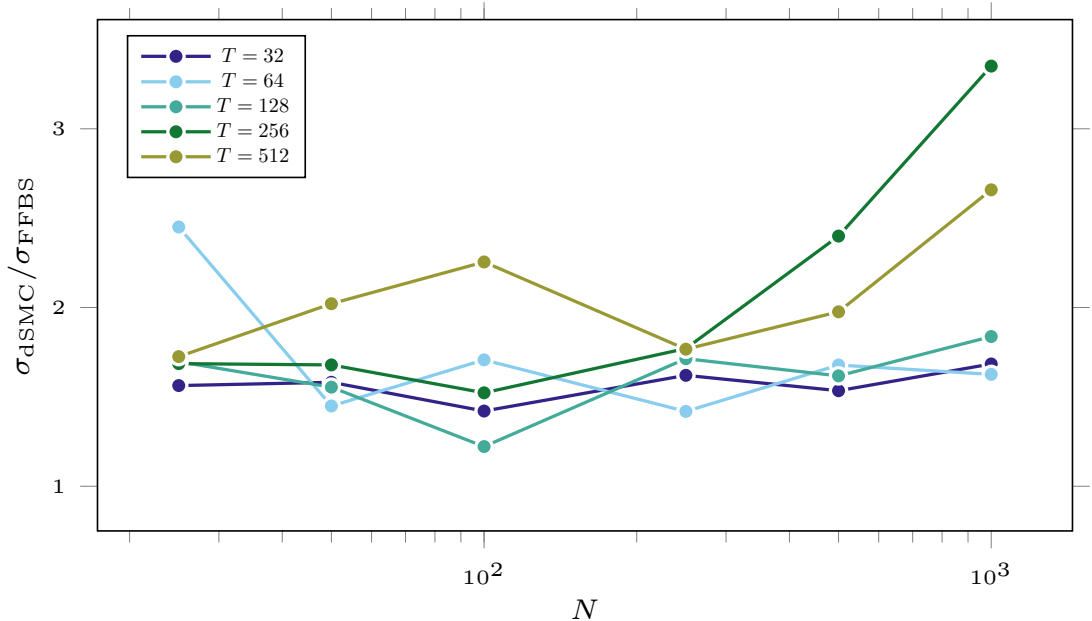
\begin{figure}[!htb]
    \centering
    \resizebox{\linewidth}{!}{
    \begin{tikzpicture}
\pgfplotsset{every tick label/.append style={font=\tiny}}

\pgfplotsset{
    legend image with text/.style={
        legend image code/.code={%
            \node[anchor=center] at (0.3cm,0cm) {#1};
        }
    },
}

\begin{axis}[
    xlabel shift = -3 pt,
    ylabel shift = -3 pt,
    height = 6cm,
    width = 10cm,
    legend style={nodes={scale=0.5, transform shape}},
    legend pos=north west,
    legend columns=1, %
    log basis x={10},
    tick align=outside,
    xlabel={\scriptsize $N$},
    xmode=log,
    ylabel={\scriptsize $\sigma_{\textrm{dSMC}} / \sigma_{\textrm{FFBS}}$},
    ymin=0.75,
]

\addplot [thick, color0, mark=*, mark size=2, mark options={solid,draw=white}] 
    table [x=N,y expr={ifthenelse(\thisrow{T}==32,\thisrow{ratio},nan)}, col sep=comma]
    {sections/experiments/data/cox-results.csv};
\addlegendentry{$T=32$}
\addplot [thick, color1, mark=*, mark size=2, mark options={solid,draw=white}] 
    table [x=N,y expr={ifthenelse(\thisrow{T}==64,\thisrow{ratio},nan)}, col sep=comma]
    {sections/experiments/data/cox-results.csv};
\addlegendentry{$T=64$}
\addplot [thick, color2, mark=*, mark size=2, mark options={solid,draw=white}] 
    table [x=N,y expr={ifthenelse(\thisrow{T}==128,\thisrow{ratio},nan)}, col sep=comma]
    {sections/experiments/data/cox-results.csv};
\addlegendentry{$T=128$}
\addplot [thick, color3, mark=*, mark size=2, mark options={solid,draw=white}] 
    table [x=N,y expr={ifthenelse(\thisrow{T}==256,\thisrow{ratio},nan)}, col sep=comma]
    {sections/experiments/data/cox-results.csv};
\addlegendentry{$T=256$}
\addplot [thick, color4, mark=*, mark size=2, mark options={solid,draw=white}] 
    table [x=N,y expr={ifthenelse(\thisrow{T}==512,\thisrow{ratio},nan)}, col sep=comma]
    {sections/experiments/data/cox-results.csv};
\addlegendentry{$T=512$}

\end{axis}

\end{tikzpicture}
    }
    \caption{Average relative error of running a sequential FFBS vs dSMC. dSMC always exhibits a higher error than FFBS, the ratio between the two increasing as $T$ increases.} \label{fig:cox-error}
\end{figure}
There therefore exists a natural trade-off between speed and precision, which can be beneficial or not depending on the application. In the next section we show that the increase in variance does not necessarily affect sampling performance in practice.

\subsection{Particle Gibbs sampling of theta-logistic model}
    \label{subsec:exp_gibbs}
    The goal of this section is to show how the c-dSMC algorithm can be used to perform particle Gibbs sampling while not reducing its performance compared to the sequential version of cSMC. In order to illustrate the properties of this PIT pGibbs algorithm, we consider the following theta-logistic state-space model:
\begin{align}
    x_0 &\sim \mathcal{N}(0, 1),\\
    x_t &= x_{t-1} + \tau_0 - \tau_1 \exp(\tau_2 x) + \epsilon_t, \quad \epsilon_t \sim \mathcal{N}(0, q^2), t \geq 1 \\
    y_t &= x_t + \gamma_t, \quad \gamma_t \sim \mathcal{N}(0, r^2),\quad t \geq 0.
\end{align}

This model was originally proposed by \citet{lande2003stochastic} in order to model population dynamics and has been used as a benchmark for PMCMC methods in, for example, \cite{peters2010ecological}, \citet[Chap. 16]{Chopin2020Book}. We use the same prior and data (nutria, $T+1=120$) as in these references. 

For c-dSMC, we take $q_t = \nu_t$ to be a ``locally adapted" $q_t = p_{EKS}(x_t \mid y_{0:T}, \tau_0, \tau_1, \tau_2, q, r)$ given by the parallel extended Kalman smoother described in Section~\ref{subsec:pit-gaussian}. More precisely, given an initial sample from the prior $p(\tau_0, \tau_1, \tau_2, q, r)$, we compute the iterated EKS solution with $25$ iterations and take the $q_t$'s to be the resulting approximated smoothing marginal. For all subsequent steps, given new parameters, we run a single step of the iterated EKS, starting from the previous iterated EKS approximation, and use the updated Gaussian approximated smoothing marginals as our new proposal distributions $q_t$'s.

\begin{figure}[!htb]
    \centering
    \resizebox{\linewidth}{!}{
    \begin{tikzpicture}[scale=1.]
\pgfplotsset{every tick label/.append style={font=\tiny}}

\begin{axis}[
    ylabel shift = -1 pt,
    height = 8cm,
    width = 10cm,
    xlabel = {\scriptsize Time step ($t$)},
    ylabel = {\scriptsize Update rate},
    ymin=0, ymax=1,
    xmin=0, xmax=120,
    yscale=0.5,
]
    \addplot[black, dashed, line width=0.5pt] table [x=time, y=Sequential, col sep=comma]{sections/experiments/data/update_rate.csv}; \label{update_rate_seq}
    \addplot[black, line width=0.5pt] table [x=time, y=Parallel, col sep=comma]{sections/experiments/data/update_rate.csv}; \label{update_rate_par}
\end{axis}
\end{tikzpicture}
    }
    \caption{Average update rate of the star trajectory $X^*_t$ for each time $t$.
    While the average update rate of cSMC with backward sampling
(\ref{update_rate_seq}) is often higher than that of c-dSMC (\ref{update_rate_par}), the latter is more homogeneous across time steps.}
    \label{fig:update-rate-gibbs}
\end{figure}
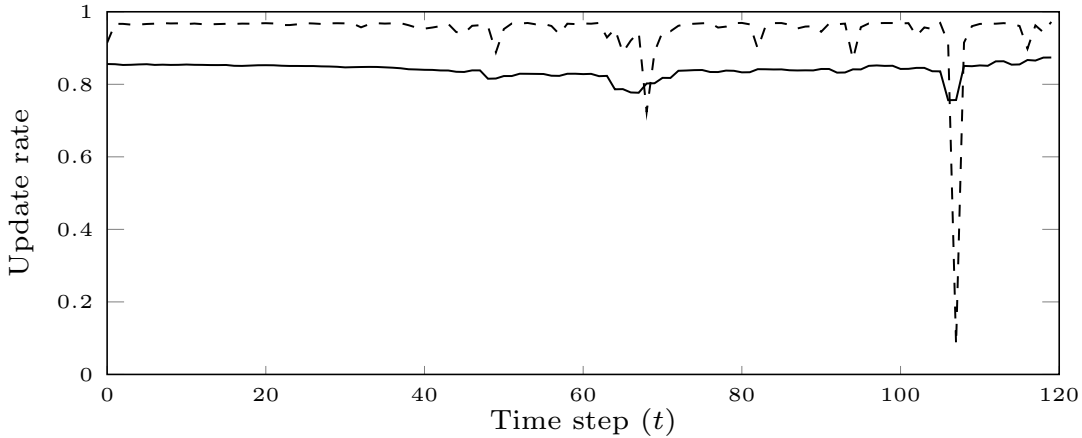

As can be seen in Figure~\ref{fig:update-rate-gibbs}, the update rate for the sampled trajectory is approximately 85\% for c-dSMC, homogeneously across all time steps \emph{without any explicit backward sampling step}. This is to be compared with the non-uniform renewal rates  ($\approx 90\%$)  of the standard pGibbs algorithm when a backward sampling step
\citep{Whiteley_disc_PMCMC, Lindsten2012} is implemented. 

Moreover, obtaining $10^5$ samples from the Gibbs chain took 480 seconds with c-dSMC, while it took 4,083 seconds for the bootstrap cSMC with backward sampling. Finally, the ACFs (auto-correlation functions) of the Markov chains formed by the parameters posterior samples are virtually identical, as illustrated by Figure~\ref{fig:acf-gibbs}.
\begin{figure}[!htb]
    \centering
    \resizebox{\linewidth}{!}{
    \begin{tikzpicture}[scale=1.]
\pgfplotsset{every tick label/.append style={font=\large}}

\begin{groupplot}[
    group style={columns=3,rows=2,horizontal sep=0.5 cm,vertical sep=1.4 cm},
    ymin=-0.05, ymax=1,
    xmin=0, xmax=100,
]   
    \nextgroupplot[title={\Large $\tau_0$}, xtick={}, xticklabels={}, ylabel={\Large Correlation}]
    \addplot[gray, dashed, line width=1pt] table [x=lag, y=Sequential, col sep=comma]{sections/experiments/data/acf-tau0.csv}; \label{acf-seq}
    \addplot[black, line width=1pt] table [x=lag, y=Parallel, col sep=comma]{sections/experiments/data/acf-tau0.csv}; \label{acf-par}
    
    \nextgroupplot[title={\Large $\tau_1$}, xtick={}, xticklabels={}, ytick={}, yticklabels={}]
    \addplot[gray, dashed, line width=1pt] table [x=lag, y=Sequential, col sep=comma]{sections/experiments/data/acf-tau1.csv};
    \addplot[black, line width=1pt] table [x=lag, y=Parallel, col sep=comma]{sections/experiments/data/acf-tau1.csv};
    
    \nextgroupplot[title={\Large $\tau_2$}, xtick={}, xticklabels={}, ytick={}, yticklabels={}]
    \addplot[gray, dashed, line width=1pt] table [x=lag, y=Sequential, col sep=comma]{sections/experiments/data/acf-tau2.csv};
    \addplot[black, line width=1pt] table [x=lag, y=Parallel, col sep=comma]{sections/experiments/data/acf-tau2.csv};
    
    \nextgroupplot[title={\Large $1/\sigma_X^2$}, xlabel={\Large Lag}, ylabel={\Large Correlation}]
    \addplot[gray, dashed, line width=1pt] table [x=lag, y=Sequential, col sep=comma]{sections/experiments/data/acf-x_prec.csv};
    \addplot[black, line width=1pt] table [x=lag, y=Parallel, col sep=comma]{sections/experiments/data/acf-x_prec.csv};
    
    \nextgroupplot[title={\Large $1/\sigma_Y^2$}, ytick={}, yticklabels={}, xlabel={\Large Lag},]
    \addplot[gray, dashed, line width=1pt] table [x=lag, y=Sequential, col sep=comma]{sections/experiments/data/acf-y_prec.csv};
    \addplot[black, line width=1pt] table [x=lag, y=Parallel, col sep=comma]{sections/experiments/data/acf-y_prec.csv};
    
    \nextgroupplot[title={\Large $X_0$}, ytick={}, yticklabels={}, xlabel={\Large Lag},]
    \addplot[gray, dashed, line width=1pt] table [x=lag, y=Sequential, col sep=comma]{sections/experiments/data/acf-x0.csv};
    \addplot[black, line width=1pt] table [x=lag, y=Parallel, col sep=comma]{sections/experiments/data/acf-x0.csv};
\end{groupplot}

\end{tikzpicture}
    }
    \caption{Auto-correlation plots for the parameters $\tau_0$, $\tau_1$, $\tau_2$, $1/\sigma_X^2$, $1/\sigma_Y^2$, and the initial state $X_0$ posterior samples. Using cSMC with backward sampling (\ref{acf-seq}) or c-dSMC (\ref{acf-par}) results in similar auto-correlation functions for the posterior samples.}
    \label{fig:acf-gibbs}
\end{figure}
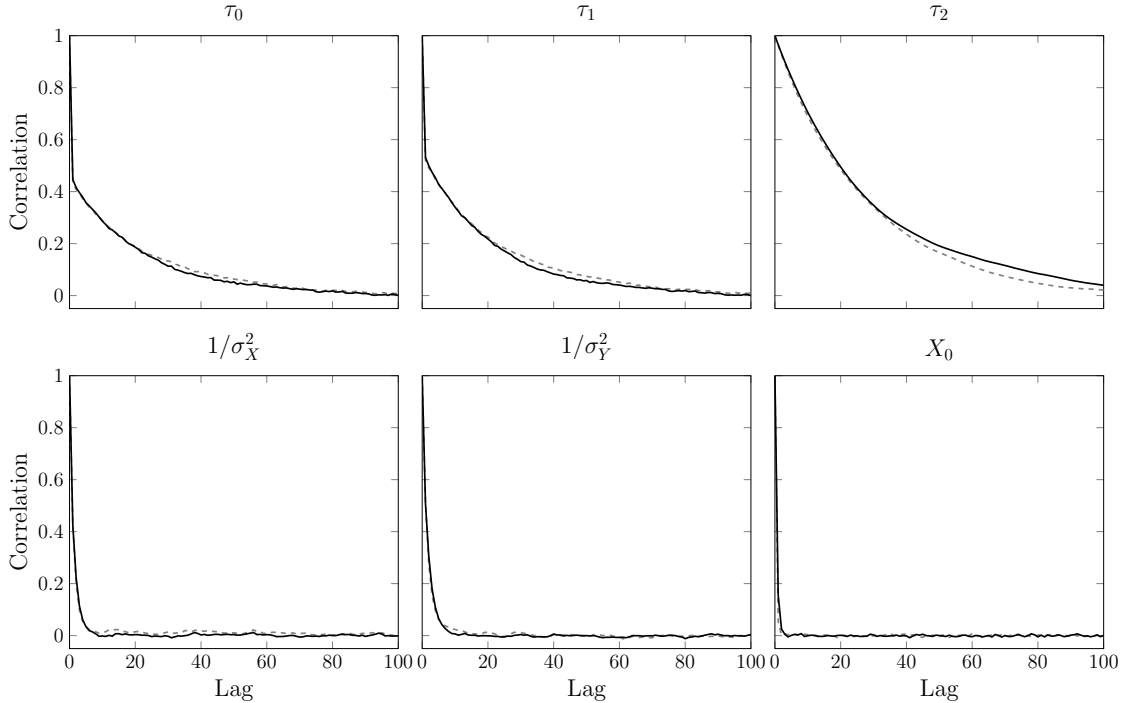

\subsection{Speed-up and variance reduction via lazy resampling} 
    \label{subsec:lazy-experiment}
    We now show how the lazy resampling methods introduced in Section~\ref{subsec:parallel-resampling} can help speed up dSMC significantly, while at the same time retaining the same variance as the original method. In order to do so, similarly to \citet[Section 4.1]{Deligiannidis2020ensemble}, we consider a constrained random walk model studied, for example, in \citet{DelMoral2004particle} and \citet{Adorisio2018exact}. While, contrarily to these works, we are not concerned with exact simulation, this model is helpful in understanding the impact of the weights variance on the total runtime and variance of dSMC with lazy resampling. Indeed, the model is controlled by a single parameter $\sigma$ which represents the noise of the constrained random walk, and directly impacts the variance of the weights in dSMC. Furthermore, this model is not easily approximated by an LGSSM, and therefore, the variance reduction method of Section~\ref{subsec:pit-gaussian} does not apply here. 

Formally, the model is defined as follows: 
\begin{equation}
    \begin{split}
        x_0 &\sim \mathcal{N}(0, 1) \\
        x_{t} &= x_{t-1} + \sigma \epsilon_{t-1}, \quad \epsilon_{t-1} \sim \mathcal{N}(0, 1)
    \end{split}
\end{equation}
and we want to sample from $p(x_{0:T} \mid -1 \leq x_t \leq 1, t=0, \ldots, T)$.
This model corresponds to the transition kernel $P_t(\dd{x}_t \mid x_{t-1}) \sim \mathcal{N}(x_{t-1}, \sigma^2)$ with potential function $h_t(x_t) = \mathbbm{1}_{[-1, 1]}(x_t)$. Following \citet{Deligiannidis2020ensemble}, we consider the proposal $q_t = \mathcal{U}([-1, 1])$; the weights $\omega_t$ are then upper-bounded by $(2 \pi \sigma)^{-1/2}$. As $\sigma$ gets higher, we expect the lazy resampling schemes in Section~\ref{subsec:parallel-resampling} to perform better. 
In order to compare the different smoothers alternatives, we estimate the Fisher score $\mathbb{E}\left[\varphi(X_{0:T}) \mid -1 \leq X_t \leq 1, t=0, \ldots, T\right]$ (up to a multiplicative constant) of this model, where
\begin{align}\label{eq:fisher_rare}
    \varphi(x_{0:T}) = \log(\sigma) + \frac{1}{\sigma^3} \sum_{t=1}^T\left(x_t - x_{t-1}\right)^2.
\end{align}

For the sake of simplicity, we only consider the rejection version of our lazy resampling methods. This is because the Metropolis--Hastings version is biased, so that convergence theorems do not apply, and because it has already been proven to  work better than its rejection counterpart (when one is not worried about the unbiasedness of the resulting algorithm) in \citet{Murray2016ParallelResampling}. 

In Figure~\ref{fig:runtime-lazy}, we take $\sigma$ to be in $\{0.3, 0.4, 0.5\}$, this set being taken to be around the value when using lazy resampling starts to outperform FFBS, and we report the average run times of FFBS, dSMC with systematic resampling (sys-dSMC), and dSMC with rejection-resampling (rs-dSMC), together with the respective the variance of the resulting Fisher's score estimates. 

\begin{figure}
    \centering
    \resizebox{\textwidth}{!}{
    \input{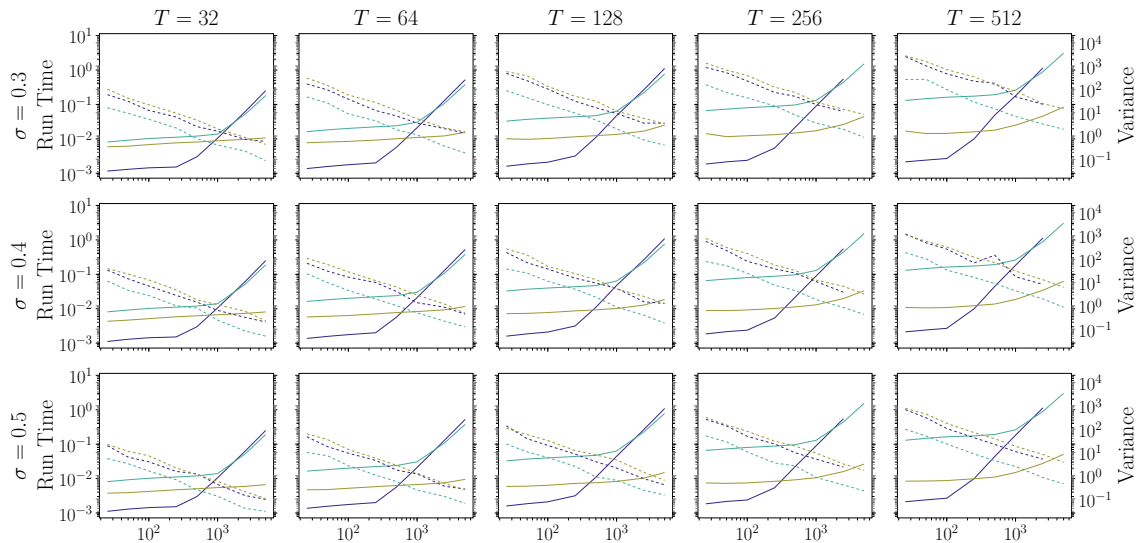}
    }
    \caption{Average run times (\ref{line:runtime-par}) to compute, and variances (\ref{line:variance-par}) of, Fisher's score estimate \eqref{eq:fisher_rare} as a function of $N$ for FFBS (\ref{line:runtime-FFBS}), sys-dSMC (\ref{line:runtime-par}), and rs-dSMC (\ref{line:runtime-lazy}), and different values of $T$ and $\sigma$. We can see that for lower variance weights regimes (higher $\sigma$), rs-dSMC runs largely faster than both FFBS and sys-dSMC, while at the same time retaining roughly the same variance as sys-dSMC.}
    \label{fig:runtime-lazy}
\end{figure}

For low $N$'s, sys-dSMC is the fastest, with fairly high variance estimates of the Fisher score, as previously discussed in Section~\ref{subsec:ffbs_comp}. However, for larger $N$ values, despite its random run time, rs-dSMC completely outperforms both FFBS and sys-dSMC in terms of speed. Moreover, for $\sigma=0.4$ and $\sigma=0.5$ and all $T$'s, the slowest running rs-dSMC ($N=5000$) is faster than the fastest running FFBS ($N=25$) and exhibits a lower Fisher score estimate variance than FFBS with more particles than $N=25$. Finally, this improved performance becomes better as the number of time steps $T$ increases, therefore confirming the appeal of dSMC for high values of $T$.

\section{Discussion}
\label{sec:discussion}
In this article we have introduced de-Sequentialized \MC{}, the first fully parallel-in-time particle smoother. This algorithm exhibits $\mathcal{L}_p$ error bounds that scale polynomially in the number of times steps and inverse proportionally to the number of particles used. Futhermore, we have shown how one can build a conditional version of dSMC, to be used, for example, in particle Gibbs algorithm. Furthermore, we discussed two variance reduction schemes based on parallel-in-time linear Gaussian state-space models approximants, as well as lazy resampling schemes. The resulting algorithms have then be shown to be competitive with standard sequential methods in different non-trivial regimes.

While the Gaussian approximations recover a lot of practical use cases, their nature makes them inadequate to approximate, for example, multi-modal posteriors. Designing proposals with more modeling capacity, anf fully utilizing the additional degree of freedom offered by the different roles of $\nu$ and $q$ is an important direction of future work. This could be done, for instance, using direct gradient methods~\citep{Corenflos2021DPF,Naesseth2017DPF,Maddison2017DPF,Le2017DPF} or more iterative methods~\citep{Guarniero2017Iterated,heng2020controlled}. 

Our parallel smoother exhibits good statistical and computational properties in non-trivial regimes, and allows faster inference at the cost of some precision. The loss of precision coming from the need to use independent proposal distributions, and we believe future research should maybe directed towards using pathwise proposals instead, for example by further leveraging the LGSSM approximants of \citet{Yaghoobi2021ParallelEKF}.

Because we developed a conditional version of dSMC, our algorithm can be used \textit{mutatis mutandis} within the unbiased coupled smoothing framework of \citet{jacob2019smoothing}. While (non-lazy) dSMC exhibits higher variance than its sequential counterparts (for the same number of particles), the framework of \citet{jacob2019smoothing} allows to average independent such estimates to increase the precision of the resulting estimate arbitrarily, making the gain of speed particularly attractive in this context.

An important technical limitation of our methodology is the necessity, at each level of the recursion, to explicitly form several $N \times N$ matrices. While this does not impact the theoretical logarithmic properties of our algorithm, this clearly limits the number of particles that we can use in at least two way: the memory footprint will scale quadratically with it, and the number of threads being limited, a processing bottleneck may appear (as illustrated in Figure~\ref{fig:cox-runtime}). We mitigated these issues by utilizing the parallel resampling perspective of \citet{Murray2016ParallelResampling} as a lazy resampling scheme, never computing more than $N$ weights at once, which allowed us to improve the scalability of dSMC in the low weights variance regime. We believe that this method can be further improved by using non-uniform proposals on the indices pairs $(I, J)$ so as to target specific pairs that have a higher \textit{a priori} chance of resulting in a high weight. It was also suggested in \citet{corenflos2022coupled} that using ensemble techniques in parallel resampling schemes may result in an improved performance at the cost of a slightly higher memory consumption. Both these extensions deserve more investigation.

On the computational resource perspective, over the years parallel processing hardware have continually increased both the memory and number of threads, so we except our algorithm to become increasingly competitive in the future. Similarly, it is also possible to distribute the computations across several processors (be it GPUs of CPUs), which in turn would result in making the algorithm scale better with the number of time steps or particles, provided that the communication cost between processors remains limited. Combining this technical solution with the lazy resampling approach of Section~\ref{subsec:parallel-resampling} in particular would likely result in a very competitive smoothers.

Finally, it was recently suggested in \citet{Deligiannidis2020ensemble} that it is possible to perform perfect sampling of SSMs smoothing distributions provided we use independent proposals. While our algorithm does not sample exactly from the same proposal distribution, it is our hope that the methods developed here could be applied to sampling from their proposal distribution too, thereby making their sampling algorithm scale logarithmically in time.

%\section*{Acknowledgments}
%\input{sections/acks}

\section*{Individiual contributions}
The original idea for this article comes from discussions between Adrien Corenflos and Simo S\"arkk\"a. The methodology of dSMC was developed by Adrien Corenflos in collaboration with Nicolas Chopin. The pGibbs and lazy resampling extensions are both due to Adrien Corenflos while the LGSSM approximants are jointly due to Simo S\"arkk\"a and Adrien Corenflos. The original proofs of this article's results are due to Adrien Corenflos, the convergence rate of Proposition~\ref{prop:consistency} being subsequently improved with the help of Nicolas Chopin. The experimental results are all due to Adrien Corenflos. The first version of this article was written by Adrien Corenflos, after which all authors contributed to the writing.

\bibliographystyle{apalike}
\bibliography{main.bib}

\newpage
\appendix
\section{Parallel combination algorithm}
\label{app:parallel-combination}
We now reproduce a parallel equivalent to Algorithm~\ref{algo:recursion}. It can generally be thought of a divide-and-conquer algorithm akin to prefix-sum algorithms, but not requiring associativity of the operator. Algorithm~\ref{algo:parallel-reshape-combination} is phrased in terms of generic operators and elements which, in the particular case of parallel particle smoothing, need to be taken to be, respectively, the operator defined in Algorithm~\ref{algo:block-combination} and the set of particles, weights and partial normalizing constants.

\begin{algorithm}[H]
\SetAlgoLined
\KwResult{Combined array}
\Fn{\textsc{ParallelReshapeCombination}$\left(Z_{1:K}, \textsc{Operator}\right)$}{
    Find $L$ such that $2^{L-1} < K \leq 2^L$\;
    \ParFor{$t=1, \ldots, K$}{
        \tcp{Flag that says if we should use the value or not}
        $b_t \leftarrow 1$\; 
    }
    \tcp{Pad $Z$ to the next power of $2$ using some \textsc{null} value}
    \ParFor{$t=K+1, \ldots, 2^L$}{
        $Z_t \leftarrow \textsc{null}$\;
        $b_t \leftarrow 0$\;
    }
    \For{$l=0, \ldots L-1$}{
        \ParFor{$n=1, \ldots 2^{L-l}$}{
            \tcc{Join the $Z$'s block by block, this corresponds to reshaping the array and do not result in creating a new array.}
            $Y_n \leftarrow \left[Z_{1 + (n-1) 2^l}, Z_{2 + (n-1) 2^l}, \ldots Z_{n 2 ^l}\right]$\;
        }
        \ParFor{$n=1, \ldots 2^{L-l}$}{
            \tcc{Combine the adjacent odd and even $Y$'s if we have not reached the padding threshold, otherwise, just leave the data unchanged.}
            \If{$b_{1 + n 2^l} = 1$}{
                $\left[Z_{1 + (n-1) 2^l}, Z_{2 + (n-1) 2^l}, \ldots Z_{n 2 ^l}\right], \left[Z_{1 + n 2^l}, Z_{2 + n 2^l}, \ldots Z_{(n+1) 2 ^l}\right] \leftarrow \textsc{CombinationOperator}\left(Y_n, Y_{n+1}\right)$\;
            }
        }
    }
\Return{$Z_{1:K}$}
}
\caption{Generic parallel combination via array reshaping}
\label{algo:parallel-reshape-combination}
\end{algorithm}

It is worth noting that Algorithm~\ref{algo:parallel-reshape-combination} and Algorithm~\ref{algo:recursion} are not strictly equivalent. This is because the combination operator used for smoothing is random and depends on the state of a random number generator. In fact two reasons make these two algorithm differ:
\begin{enumerate}
    \item The order in which the nodes at a given depth of Algorithm~\ref{algo:recursion} are handled is arbitrary. Similarly for the order in which we combine adjacent blocks in Algorithm~\ref{algo:parallel-reshape-combination}.
    \item The splitting of Algorithm~\ref{algo:parallel-reshape-combination}, although corresponding to a balanced tree, will not correspond to the mid-point splitting of Algorithm~\ref{algo:recursion} except when $T+1$ is a power of $2$.
\end{enumerate}
However, both algorithms are consistent and can be analyzed by Propositon~\ref{prop:consistency} in the same way.

\section{Lazy resampling algorithms}
\label{app:lazy-resampling}
We now describe the lazy resampling algorithms introduced in Section~\ref{subsec:parallel-resampling}. The Metropolis-Hastings version is given by Algorithm~\ref{algo:m-h-resampling}, while the rejection sampling one is given by Algorithm~\ref{algo:r-s-resampling}.

\begin{algorithm}[!htb]
\SetAlgoLined
\KwResult{Resampling indices $(I_m, J_m)$ for $m=1, \ldots, N$.}
\Fn{\textsc{MHResampling}$\left(X_{c-1}^{1:N}, X_{c}^{1:N}, \omega_c, B\right)$}{
    \ParFor{$m=1, \ldots, N$}{
        $(I_m, J_m) \leftarrow (m, m)$\;
        \For{$b=1, \ldots, B$}{
            Sample $u \sim \mathcal{U}([0, 1])$\;
            Sample $I^*, J^* \sim \mathcal{U}(\{1, \ldots, N\})$, independently\;
            \If{$u < \omega_c(X_{c-1}^{I^*}, X_{c}^{J^*}) / \omega_c(X_{c-1}^{I_m}, X_{c}^{J_m})$}{
                $(I_m, J_m) \leftarrow (I^*, J^*)$\;
            }
        }
    }
    \Return{$I_{1:N}, J_{1:N}$}
}
\caption{Metropolis-Hastings lazy resampling algorithm}
\label{algo:m-h-resampling}
\end{algorithm}

It is worth noting that, contrarily to Algorithm 3 (Code 3) in \citet{Murray2016ParallelResampling}, the initial proposal in Algorithm~\ref{algo:r-s-resampling} is random and not deterministic. This is because the deterministic starting point of \citet{Murray2016ParallelResampling} would result in a bias when subsampling $N$ candidates from the $N \times N$ entries in the weight matrix.

\begin{algorithm}[!htb]
\SetAlgoLined
\KwResult{Resampling indices $(I_m, J_m)$ for $m=1, \ldots, N$.}
\Fn{\textsc{RSResampling}$\left(X_{c-1}^{1:N}, X_{c}^{1:N}, \omega_c, \overline{\omega}_c\right)$}{
    \tcp{$\overline{\omega}_c$ is such that $\omega_c(x,y) \leq \overline{\omega}_c$ for all $x,y$.}
    \ParFor{$m=1, \ldots, N$}{
        Sample $I_m$, $J_m \sim \mathcal{U}(\{1, \ldots, N\})$, independently\;
        Sample $u \sim \mathcal{U}([0, 1])$\;
        \While{$u > \omega_c(X_{c-1}^{I_m}, X_{c}^{J_m}) / \overline{\omega}_c$}{
            Sample $I_m, J_m \sim \mathcal{U}(\{1, \ldots, N\})$, independently\;
        }
    }
    \Return{$I_{1:N}, J_{1:N}$}
}
\caption{Rejection-sampling lazy resampling algorithm}
\label{algo:r-s-resampling}
\end{algorithm}

\section{Proof of Proposition~\ref{prop:consistency}}
\label{app:proof}
For simplicity we only consider the case when $w_{c-1}^n = w_{c}^n= 1/N$, for all $n \in 1:N$. The general case follows from the same lines. Using Minkowski's inequality, we have
\begin{equation}
     \bbE{}\left[\left|\ft{\bbQ{}}{a}{b}(\varphi) - \ft{\bbQ{}^{N}}{a}{b}(\varphi)\right|^p\right]^{1/p}
        \leq \bbE{}\left[\left|\ft{\bbQ{}}{a}{b}(\varphi) - \ft{\widetilde{\bbQ}^N}{a}{b}(\varphi)\right|^p\right]^{1/p} + \bbE{}\left[\left|\ft{\widetilde{\bbQ}^N}{a}{b}(\varphi) - \ft{\bbQ{}^{N}}{a}{b}(\varphi)\right|^p\right]^{1/p}. \label{eq:split_equation}
\end{equation}
The second term of \eqref{eq:split_equation}, corresponding to the resampling error, can be controlled as a Monte Carlo error via \citet[][Lemma 7.3.3]{DelMoral2004Book}. Indeed, let us first notice that we have $\bbE{}\left[\ft{\bbQ^{N}}{a}{b}(\varphi) \mid \ft{X^{1:N}}{a}{b}\right] = \sum_{m,n=1}^N W^{mn}_c \varphi\left(\ft{X^{m}}{a}{c-1}, \ft{X^{n}}{c}{b}\right)$, and that, given that we are considering the multinomial resampling case, conditionally on $\ft{X^{1:N}}{a}{b}$, the variables $(l_n, r_n)_{n=1}^N$ are independent. In this case, 
\begin{align}
    \bbE{}\left[\left|\ft{\widetilde{\bbQ}^N}{a}{b}(\varphi) - \ft{\bbQ{}^{N}}{a}{b}(\varphi)\right|^p \mid \ft{X^{1:N}}{a}{b} \right]^{1/p} \leq d(p) \frac{\norm{\varphi}_{\infty}}{N^{1/2}}
\end{align}
for some constant $d(p) \leq 2^{(p+1)/p}$, so that the tower law ensures that 
\begin{align}
    \bbE{}\left[\left|\ft{\widetilde{\bbQ}^N}{a}{b}(\varphi) - \ft{\bbQ{}^{N}}{a}{b}(\varphi)\right|^p\right]^{1/p} \leq 2^{(p+1)/p} \frac{\norm{\varphi}_{\infty}}{N^{1/2}}
 \end{align}
is verified too.

On the other hand, the first term of \eqref{eq:split_equation}, corresponding to the self-normalization error, requires more attention. In order to simplify notations, let us introduce the following quantities:
\begin{align}
    \ft{\widehat{\bbQ}^N}{c}{b}(\varphi) &\coloneqq \frac{1}{N}\sum_{n=1}^N\int \bar{\omega}_c\left(X^n_{c-1}, x_c\right) \varphi\left(\ft{X^n}{a}{c-1}, \ft{x}{c}{b}\right) \ft{\bbQ}{c}{b}(\dd{\ft{x}{c}{b}}), \\ \ft{\breve{\bbQ}^{N}}{a}{b}(\varphi) &\coloneqq \frac{1}{N^2}\sum_{m,n=1}^N \bar{\omega}_c\left(X^{m}_{c-1}, X^n_c\right) \varphi\left(\ft{X^{m}}{a}{c-1}, \ft{X^n}{c}{b}\right).
\end{align}
Using Minkowski's inequality again, twice, we can now decompose the first term of \eqref{eq:split_equation} as 
\begin{align}
    \bbE{}\left[\left|\ft{\bbQ}{a}{b}(\varphi) - \ft{\widetilde{\bbQ}^N}{a}{b}(\varphi)\right|^p\right]^{1/p} \leq \bbE{}\left[\left|\ft{\bbQ}{a}{b}(\varphi) - \ft{\breve{\bbQ}^{N}}{a}{b}(\varphi)\right|^p\right]^{1/p} +  \bbE{}\left[\left|\ft{\breve{\bbQ}^{N}}{a}{b}(\varphi)- \ft{\widetilde{\bbQ}^N}{a}{b}(\varphi)\right|^p\right]^{1/p}, \label{eq:split-equation-again}
\end{align}
so that, splitting once more, we have
\begin{equation}
    \begin{split}
    \bbE{}\left[\left|\ft{\bbQ}{a}{b}(\varphi) - \ft{\breve{\bbQ}^{N}}{a}{b}(\varphi)\right|^p\right]^{1/p}
        & \leq \bbE{}\left[\left|\ft{\bbQ}{a}{b}(\varphi) - \ft{\widehat{\bbQ}^N}{c}{b}(\varphi)\right|^p\right]^{1/p}\\
        &+ \bbE{}\left[\left|\ft{\widehat{\bbQ}^N}{c}{b}(\varphi) - \ft{\breve{\bbQ}^{N}}{a}{b}(\varphi)\right|^p\right]^{1/p}. 
        \end{split}
        \label{eq:square-mc-error}
\end{equation}
Let us first remark that 
\begin{align}
  \ft{\bbQ}{a}{b}(\varphi) =
  \ft{\bbQ}{a}{c-1}\left(\ft{x}{a}{c-1} \mapsto \int \bar{\omega}_c\left(x_{c-1}, x_c\right) \varphi\left(\ft{x}{a}{c-1}, \ft{x}{c}{b}\right) \ft{\bbQ}{c}{b}(\dd{\ft{x}{c}{b}})\right).
\end{align}
The integrand $\ft{x}{a}{c-1} \mapsto \int \bar{\omega}_c\left(x_{c-1},
x_c\right) \varphi\left(\ft{x}{a}{c-1}, \ft{x}{c}{b}\right)
\ft{\bbQ}{c}{b}(\dd{\ft{x}{c}{b}})$ is upper bounded by
$\norm{\bar{\omega}_c}_{\infty} \norm{\varphi}_{\infty}$, so that we can apply the
recursion hypothesis to get
\begin{align}
    \bbE{}\left[\left|\ft{\bbQ}{a}{b}(\varphi) - \ft{\widehat{\bbQ}^{N}}{c}{b}(\varphi)\right|^p\right]^{1/p} \leq \ft{C^p}{a}{c-1} \norm{\bar{\omega}_{c}}_{\infty} \frac{\norm{\varphi}_\infty}{N^{1/2}}.
\end{align}
On the other hand, using the tower law, the second term of~\eqref{eq:square-mc-error} becomes
\begin{align}
    \bbE{}\left[\left|\ft{\widehat{\bbQ}^{N}}{c}{b}(\varphi) - \ft{\breve{\bbQ}^{N}}{a}{b}(\varphi)\right|^p\right] 
        &= \bbE\left[\bbE{}\left[\left|\ft{\widehat{\bbQ}^{N}}{c}{b}(\varphi) - \ft{\breve{\bbQ}^{N}}{a}{b}(\varphi)\right|^p \mid \ft{X^{1:N}}{a}{c-1}\right] \right]. 
\end{align}
Noting that \[\ft{\widehat{\bbQ}^{N}}{c}{b}(\varphi)
= \ft{\bbQ}{c}{b}\left(\ft{x}{c}{b} \mapsto N^{-1} \sum_{n=1}^N
  \bar{\omega}_c\left(X^n_{c-1}, x_c\right) \varphi\left(\ft{X^n}{a}{c-1},
    \ft{x}{c}{b}\right)\right),\]
and that, for all $n=1, \ldots, N$ and all $\ft{x}{c}{b}$, $N^{-1}
\sum_{n=1}^N \bar{\omega}_c\left(X^n_{c-1}, x_c\right) \varphi\left(\ft{X^n}{a}{c-1},
\ft{x}{c}{b}\right) \leq \norm{\bar{\omega}_c}_\infty \norm{\varphi}_\infty$, we can
leverage the recursion hypothesis one more time to obtain
\begin{align}
    \bbE{}\left[\left|\ft{\widehat{\bbQ}^{N}}{c}{b}(\varphi) -
\ft{\breve{\bbQ}^{N}}{a}{b}(\varphi)\right|^p \mid
\ft{X^{1:N}}{a}{c-1}\right]^{1/p}\leq \ft{C^p}{a}{c-1}
\norm{\bar{\omega}_{c}}_{\infty} \frac{\norm{\varphi}_\infty}{N^{1/2}}
\end{align}
and, applying the tower law again,
\begin{align}
    \bbE{}\left[\left|\ft{\widehat{\bbQ}^{N}}{c}{b}(\varphi) - \ft{\breve{\bbQ}^{N}}{a}{b}(\varphi)\right|^p\right] ^{1/p}
        &\leq \ft{C^p}{a}{c-1}\norm{\bar{\omega}_{c}}_{\infty} \frac{\norm{\varphi}_\infty}{N^{1/2}}.
\end{align}
This ensures that 
\begin{align}
    \bbE{}\left[\left|\ft{\bbQ}{a}{b}(\varphi) - \ft{\breve{\bbQ}^{N}}{a}{b}(\varphi)\right|^p\right]^{1/p}
        & \leq 2 \ft{C^p}{a}{c-1}\norm{\bar{\omega}_{c}}_{\infty} \frac{\norm{\varphi}_\infty}{N^{1/2}}.
\end{align}
Similarly, instead of introducing $\ft{\widehat{\bbQ}^{N}}{c}{b}(\varphi)$,
we could have introduced the similar quantity 

\[\ft{\widehat{\bbQ}^{N}}{a}{c-1}(\varphi) = \frac 1 N \sum_{n=1}^N\int
  \bar{\omega}_c\left(x_{c-1}, X^n_c\right)
  \varphi\left(\ft{x}{a}{c-1}, \ft{X^n}{c}{b}\right)
  \ft{\bbQ}{a}{c-1}(\dd{\ft{x}{a}{c-1}})\]
to obtain:
\begin{align}
    \bbE{}\left[\left|\ft{\bbQ}{a}{b}(\varphi) - \ft{\breve{\bbQ}^{N}}{a}{b}(\varphi)\right|^p\right]^{1/p}
        & \leq 2 \ft{C^p}{c}{b}\norm{\bar{\omega}_{c}}_{\infty} \frac{\norm{\varphi}_\infty}{N^{1/2}}.
\end{align}
This finally ensures that 
\begin{align}
    \bbE{}\left[\left|\ft{\bbQ}{a}{b}(\varphi) - \ft{\breve{\bbQ}^{N}}{a}{b}(\varphi)\right|^p\right]^{1/p}
      & \leq 2 \min(\ft{C^p}{a}{c-1}, \ft{C^p}{c}{b})\norm{\bar{\omega}_{c}}_{\infty} \frac{\norm{\varphi}_\infty}{N^{1/2}}.\label{eq:err}
\end{align}
Now the term $\bbE{}\left[\left|\ft{\breve{\bbQ}^{N}}{a}{b}(\varphi)- \ft{\widetilde{\bbQ}^N}{a}{b}(\varphi)\right|^p\right]^{1/p}$ can be controlled in a way similar to the one used in \cite[][Lemma 11.2]{Chopin2020Book}. Indeed we first note that $\ft{\widetilde{\bbQ}^N}{a}{b}(\varphi) - \ft{\breve{\bbQ}^N}{a}{b}(\varphi) = \ft{\widetilde{\bbQ}^N}{a}{b}(\varphi) \left(1 - \ft{\breve{\bbQ}^N}{a}{b}(1)\right)$, so that 
\begin{align}
  \bbE{}\left[\left|\ft{\breve{\bbQ}^{N}}{a}{b}(\varphi)- \ft{\widetilde{\bbQ}^N}{a}{b}(\varphi)\right|^p\right]^{1/p} \leq \norm{\varphi}_{\infty} \bbE{}\left[\left|1 - \ft{\breve{\bbQ}^N}{a}{b}(1)\right|^p\right]^{1/p}
\end{align}
Moreover, $\ft{\bbQ}{a}{b}(1) = 1$ by definition, so that we can rewrite 
\begin{align}
  \bbE{}\left[\left|1 - \ft{\breve{\bbQ}^N}{a}{b}(1)\right|^p\right]^{1/p}
  &= \bbE{}\left[\left|\ft{\bbQ}{a}{b}(1) - \ft{\breve{\bbQ}^N}{a}{b}(1)\right|^p\right]^{1/p}
\end{align}
which can be bounded similarly to~\eqref{eq:err}, giving
\begin{align}
  \bbE{}\left[\left|1 - \ft{\breve{\bbQ}^N}{a}{b}(1)\right|^p\right]^{1/p} 
  & \leq 2 \min(\ft{C^p}{a}{c-1}, \ft{C^p}{c}{b}) \norm{\bar{\omega}_c}_{\infty} \frac{1}{N^{1/2}}.
\end{align}
This results in the following inequality
\begin{align}
    \bbE{}\left[\left|\ft{\breve{\bbQ}^{N}}{a}{b}(\varphi)- \ft{\widetilde{\bbQ}^N}{a}{b}(\varphi)\right|^p\right]^{1/p}
    & \leq 2 \min(\ft{C^p}{a}{c-1}, \ft{C^p}{c}{b}) \norm{\bar{\omega}_c}_{\infty} \frac{\norm{\varphi}_{\infty}}{N^{1/2}}
\end{align}
Putting everything together, we obtain
\begin{align}
    \bbE{}\left[\left|\ft{\bbQ}{a}{b}(\varphi) - \ft{\bbQ^{N}}{a}{b}(\varphi)\right|^p\right]^{1/p}
        &\leq \left(4 \min(\ft{C^p}{a}{c-1}, \ft{C^p}{c}{b})\norm{\bar{\omega}_c}_{\infty} + 2^{(p+1)/p)}\right)\frac{\norm{\varphi}_{\infty}}{N^{1/2}}.
\end{align}

\end{document}